\documentclass[a4paper,UKenglish,cleveref, autoref]{lipics-v2019}
\usepackage{latexsym,amsthm,amsmath,amssymb}

\usepackage[T1]{fontenc}
\usepackage{bbm}
\usepackage{mathtools}
\usepackage{amssymb}
\usepackage{amsmath}
\usepackage{amsthm}
\usepackage{amsfonts}
\usepackage{pifont}
\usepackage{url}
\usepackage{multirow}
\usepackage{multicol}
\usepackage{fancybox}
\usepackage{colortbl}
\usepackage{soul}
\usepackage{color}
\usepackage{tikz}
\usepackage{tkz-berge}
\usepackage{rotating}
\usepackage{pdflscape}
\usepackage{afterpage}
\usepackage{comment}
\usepackage{etoolbox}
\usepackage{environ}
\usepackage{dirtytalk}

\usepackage{cite}
\usepackage{hyperref}

\usepackage{algorithm}
\usepackage{algorithmic}

\newcommand{\probl}[3]{
  \begin{flushleft}
    \fbox{
      \begin{minipage}{13.5cm}
        \noindent {\sc #1}\\
        {\bf Instance:} #2\\
        {\bf Question:} #3
      \end{minipage}}
    \medskip
  \end{flushleft}
}

\newcommand{\paraprobl}[5]
{
  \begin{flushleft}
    \fbox{
      \begin{minipage}{#5cm}
        \noindent {\textsc {#1}}\\
        {\bf Instance:} #2\\
        {\bf Parameter:} #4\\
        {\bf Question:} #3
      \end{minipage}
    }
  \end{flushleft}
}

\newcommand{\Acal}{\mathcal{A}}

\usetikzlibrary{decorations,arrows,shapes}

\usepackage{complexity}

\newtheorem{rrule}{Reduction Rule}
\newtheorem{observation}{Observation}

\newcommand{\balpha}{\boldsymbol{\alpha}}
\newcommand{\bmd}{\boldsymbol{d}}

\newcommand{\bigO}[1]{\mathcal{O}\!\left(#1\right)}
\newcommand{\bigOs}[1]{\mathcal{O}^*\!\left(#1\right)}

\newcommand{\dcup}{\ \dot\cup\ }

\NewEnviron{sproof}[1] {
    \begin{proof}[Proof of safeness of Rule~#1]\BODY\end{proof}
}

\newcommand{\pname}[1]{\textsc{#1}}

\newcommand{\inners}{1.2pt}
\newcommand{\outers}{1pt}

\newcommand{\tdef}[1]{\emph{#1}}

\newclass{\Hard}{hard}
\newclass{\Hness}{hardness}

\newclass{\Complete}{complete}
\newclass{\para}{para}

\newclass{\Cness}{completeness}

\newfunc{\tw}{tw}
\newfunc{\YES}{YES}
\newfunc{\NOi}{NO}
\newfunc{\dc}{dc}
\newfunc{\dcc}{d\overline{c}}
\newclass{\ETH}{ETH}
\newclass{\dgr}{deg}

\title{Finding Cuts of Bounded Degree: Complexity, FPT and Exact Algorithms, and Kernelization}

\titlerunning{Finding Cuts of Bounded Degree}

\author{Guilherme C. M. Gomes}{Universidade Federal de Minas Gerais, \and Departamento de Ciência da Computação, Belo Horizonte, Brazil}{gcm.gomes@dcc.ufmg.br}{https://orcid.org/0000-0002-5164-1460
}{Coordenação de Aperfeiçoamento de Pessoal de Nível Superior - Brasil (CAPES) - Finance Code 001.}

\author{Ignasi Sau}{CNRS, LIRMM, Universit\'e de Montpellier, Montpellier France}{ignasi.sau@lirmm.fr}{https://orcid.org/0000-0002-8981-9287
}{Projects DEMOGRAPH (ANR-16-CE40-0028) and ESIGMA (ANR-17-CE23-0010).}

\authorrunning{G.\,C. M. Gomes and I.\,Sau}

\Copyright{Guilherme C. M. Gomes and Ignasi Sau}


\category{}

\relatedversion{}

\supplement{}


\acknowledgements{}

\nolinenumbers

\hideLIPIcs

\EventEditors{John Q. Open and Joan R. Access}
\EventNoEds{2}
\EventLongTitle{42nd Conference on Very Important Topics (CVIT 2019)}
\EventShortTitle{CVIT 2019}
\EventAcronym{CVIT}
\EventYear{2016}
\EventDate{December 24--27, 2016}
\EventLocation{Little Whinging, United Kingdom}
\EventLogo{}
\SeriesVolume{42}
\ArticleNo{23}

\begin{document}

\maketitle

\begin{abstract}
A \emph{matching cut} is a partition of the vertex set of a graph into two sets $A$ and $B$ such that each vertex has at most one neighbor in the other side of the cut. The \textsc{Matching Cut} problem asks whether a graph has a matching cut, and has been intensively studied in the literature. Motivated by a question posed by Komusiewicz et al.~[IPEC 2018], we introduce a natural generalization of this problem, which we call \textsc{$d$-Cut}: for a positive integer $d$, a \emph{$d$-cut} is a bipartition of the vertex set of a graph into two sets $A$ and $B$ such that each vertex has at most $d$ neighbors across the cut.  We generalize (and in some cases, improve) a number of results for the \textsc{Matching Cut} problem. Namely, we begin with an \NP-hardness reduction for \textsc{$d$-Cut} on $(2d+2)$-regular graphs and a polynomial algorithm for graphs of maximum degree at most $d+2$. The degree bound in the hardness result is unlikely to be improved, as it would disprove a long-standing conjecture in the context of internal partitions. We then give {\sf FPT} algorithms for several parameters: the maximum number of edges crossing the cut, treewidth, distance to cluster, and distance to co-cluster. In particular, the treewidth algorithm improves upon the running time of the best known algorithm for \textsc{Matching Cut}. Our main technical contribution, building on the techniques of Komusiewicz et al.~[IPEC 2018], is a polynomial kernel for \textsc{$d$-Cut} for every positive integer $d$, parameterized by the distance to a cluster graph. We also rule out the existence of polynomial kernels when parameterizing simultaneously by the number of edges crossing the cut, the treewidth, and the maximum degree. Finally, we provide an exact exponential algorithm slightly faster than the naive brute force approach running in time $\bigOs{2^n}$.
\end{abstract}

\section{Introduction}
\label{sec:intro}

 A \tdef{cut} of a graph $G = (V, E)$ is a bipartition of its vertex set $V(G)$ into two non-empty sets, denoted by $(A,B)$.
The set of all edges with one endpoint in $A$ and the other in $B$ is the \tdef{edge cut}, or the set of \tdef{crossing edges}, of $(A,B)$.
A \tdef{matching cut} is a (possibly empty) edge cut that is a matching, that is, such that its edges are pairwise vertex-disjoint. Equivalently, $(A, B)$ is a matching cut of $G$ if and only if every vertex is incident to at most one crossing edge of $(A, B)$~\cite{matching_cut_graham, chvatal_matching_cut}, that is, it has at most one neighbor across the cut.

Motivated by an open question posed by Komusiewicz et al.~\cite{matching_cut_ipec} during the presentation of their article,  we investigate a natural generalization that arises from this alternative definition, which we call \tdef{$d$-cut}.
Namely, for a positive integer $d \geq 1$, a $d$-cut is a a cut $(A, B)$ such that each vertex has at most $d$ neighbors across the partition, that is, every vertex in $A$ has at most $d$ neighbors in $B$, and vice-versa. Note that a $1$-cut is a matching cut.
As expected, not every graph admits a $d$-cut, and the \pname{$d$-Cut} problem is the problem of, for a fixed integer $d \geq 1$, deciding whether or not an input graph $G$ has a $d$-cut.


When $d=1$, we refer to the problem as \pname{Matching Cut}.
Graphs with no matching cut first appeared in Graham's manuscript~\cite{matching_cut_graham} under the name of \textit{indecomposable graphs}, presenting some examples and properties of decomposable and indecomposable graphs, leaving their recognition as an open problem.
In answer to Graham's question, Chv\'atal~\cite{chvatal_matching_cut} proved that the problem is \NP-hard for graphs of maximum degree at least four and polynomially solvable for graphs of maximum degree at most three; in fact, as shown by Moshi~\cite{matching_cut_moshi}, every graph of maximum degree three and at least eight vertices has a matching cut.

Chvátal's results spurred a lot of research on the complexity of the problem~\cite{matching_cut_ipec,matching_cut_structural,matching_cut_tcs, matching_cut_diameter, matching_cut_planar, matching_cut_series_parallel, stable_cutset_line_graphs}.
In particular, Bonsma~\cite{matching_cut_planar} showed that \pname{Matching Cut} remains \NP-hard for planar graphs of maximum degree four and for planar graphs of girth five;
Le and Randerath~\cite{stable_cutset_line_graphs} gave an \NP-hardness reduction for bipartite graphs of maximum degree four;
Le and Le~\cite{matching_cut_diameter} proved that \pname{Matching Cut} is \NP-hard for graphs of diameter at least three, and presented a polynomial-time algorithm for graphs of diameter at most two.
Beyond planar graphs, Bonsma's work~\cite{matching_cut_planar} also proves that the matching cut property is expressible in monadic second order logic and, by Courcelle's Theorem~\cite{courcelle_theorem}, it follows that \pname{Matching Cut} is $\FPT$ when parameterized by the treewidth of the input graph; he concludes with a proof that the problem admits a polynomial-time algorithm for graphs of bounded cliquewidth.

Kratsch and Le~\cite{matching_cut_tcs} noted that Chv\'atal's original reduction also shows that, unless the Exponential Time Hypothesis~\cite{eth} (\ETH) fails\footnote{The $\ETH$ states that 3-\textsc{SAT} on $n$ variables cannot be solved in time $2^{o(n)}$; see~\cite{eth} for more details.}, there is no algorithm solving \pname{Matching Cut} in time $2^{o(n)}$ on $n$-vertex input graphs.
Also in~\cite{matching_cut_tcs}, the authors provide a first branching algorithm, running\footnote{The $\bigOs{\cdot}$ notation suppresses factors that are bounded by a polynomial in the input size.} in time $\bigOs{2^{n/2}}$, a single-exponential $\FPT$ algorithm when parameterized by the vertex cover number $\tau(G)$, and an algorithm generalizing the polynomial cases of line graphs~\cite{matching_cut_moshi} and claw-free graphs~\cite{matching_cut_planar}.
Kratsch and Le~\cite{matching_cut_tcs} also asked for the existence a single-exponential algorithm parameterized by treewidth.
In response, Aravind et al.~\cite{matching_cut_structural} provided a $\bigOs{12^{\tw(G)}}$ algorithm for \pname{Matching Cut} using nice tree decompositions, along with $\FPT$ algorithms for other structural parameters, namely neighborhood diversity, twin-cover, and distance to split graph.

The natural parameter -- the number of edges crossing the cut -- has also been considered.
Indeed, Marx et al.~\cite{marx_treewidth_reduction} tackled the \pname{Stable Cutset} problem, to which \textsc{Matching Cut} can be easily reduced via the line graph, and through a breakthrough technique showed that this problem is $\FPT$ when parameterized by the maximum size of the stable cutset.
Recently, Komusiewicz et al.~\cite{matching_cut_ipec} improved on the results of Kratsch and Le~\cite{matching_cut_tcs}, providing an exact exponential algorithm for \textsc{Matching Cut} running in  time $\bigOs{1.3803^n}$, as well as $\FPT$ algorithms parameterized by the distance to a cluster graph and the distance to a co-cluster graph, which improve the algorithm parameterized by the vertex cover number, since both parameters are easily seen to be smaller than the vertex cover number.
For the distance to cluster parameter, they also presented a quadratic kernel; while for a combination of treewidth, maximum degree, and number of crossing edges, they showed that no polynomial kernel exists unless $\NP \subseteq \coNP/\poly$.

A problem  closely related to \pname{$d$-Cut} is that of \pname{Internal Partition}, first studied by Thomassen~\cite{internal_partition_thomassen}.
In this problem, we seek a bipartition of the vertices of an input graph such that every vertex has at least as many neighbors in its
own part as in the other part. Such a partition is called an \tdef{internal partition}.
Usually, the problem is posed in a more general form: given functions $a,b: V(G) \rightarrow \mathbb{Z}_+$, we seek a bipartition $(A,B)$ of $V(G)$ such that every $v \in A$ satisfies $\dgr_A(v) \geq a(v)$ and every $u \in B$ satisfies $\dgr_B(u) \geq b(u)$, where $\dgr_A(v)$ denotes the number of neighbors of $v$ in the set $A$. Such a partition is called an \tdef{$(a,b)$-internal partition}.

Originally, Thomassen asked in~\cite{internal_partition_thomassen} whether for any pair of positive integers $s,t$, a graph $G$ with $\delta(G) \geq s + t + 1$ has a vertex bipartition $(A,B)$ with $\delta(G[A]) \geq s$ and $\delta(G[B]) \geq t$.
Stiebitz~\cite{internal_partition_stiebitz} answered that, in fact, for any graph $G$ and any pair of functions $a,b: V(G) \rightarrow \mathbb{Z}_+$ satisfying $\dgr(v) \geq a(v) + b(v) + 1$ for every $v \in V(G)$, $G$ has an $(a,b)$-internal partition.
Following Stiebitz's work, Kaneko~\cite{internal_partition_triangle_free} showed that if $G$ is triangle-free, then the pair $a,b$ only needs to satisfy $\dgr(v) \geq a(v) + b(v)$.
More recently, Ma and Yang~\cite{internal_partition_c4_free} proved that, if $G$ is $\{C_4, K_4, \text{diamond}\}$-free, then $\dgr(v) \geq a(v) + b(v) - 1$ is enough.
Furthermore, they also showed, for any pair $a,b$, a family of graphs such that $\dgr(v) \geq a(v) + b(v) - 2$ for every $v \in V(G)$ that do not admit an $(a,b)$-internal partition.

It is conjectured that, for every positive integer $r$, there exists some constant $n_r$ for which every $r$-regular graph with more than $n_r$ vertices has an internal partition~\cite{DeVos09,internal_partition_regular6} (the conjecture for $r$ even appeared first in~\cite{internal_partition_regular3_4}).
The cases $r \in \{3,4\}$ have been settled by Shafique and Dutton~\cite{internal_partition_regular3_4}; the case $r=6$ has been verified by Ban and Linial~\cite{internal_partition_regular6}.
This latter result implies that every 6-regular graph of sufficiently large size has a 3-cut.

\medskip

\noindent \textbf{Our results}. We aim at generalizing several of the previously reported results  for \pname{Matching Cut}.
First, we show in Section~\ref{sec:np}, by using a reduction inspired by Chvátal's~\cite{chvatal_matching_cut}, that for every $d \geq 1$, \pname{$d$-Cut} is \NP-hard even when restricted to $(2d+2)$-regular graphs and that, if $\Delta(G) \leq d+2$, finding a $d$-cut can be done in polynomial time. The degree bound in the \NP-hardness result is unlikely to be improved: if we had an \NP-hardness result for \textsc{$d$-Cut} restricted to $(2d+1)$-regular graphs, this would disprove the conjecture about the existence of internal partitions on $r$-regular graphs~\cite{DeVos09,internal_partition_regular6,internal_partition_regular3_4} for $r$ odd, unless all the problems in $\NP$ could be solved in {\sl constant} time. We conclude the section by giving a simple exact exponential algorithm that,  for every $d \geq 1$, runs in time $\bigOs{c_d^n}$ for some constant $c_d < 2$, hence improving over the trivial brute-force algorithm running in time $\bigOs{2^n}$.

We then proceed to analyze the problem in terms of its parameterized complexity.
Section~\ref{sec:param} begins with a proof, using the treewidth reduction technique of Marx et al.~\cite{marx_treewidth_reduction}, that \pname{$d$-Cut} is $\FPT$ parameterized by the maximum number of edges crossing the cut.
Afterwards, we present a dynamic programming algorithm for \pname{$d$-Cut} parameterized by treewidth running in time $\bigOs{2^{\tw(G)+1}(d+1)^{2\tw(G) + 2}}$; in particular, for $d=1$ this algorithm runs in time $\bigOs{8^{\tw(G)}}$ and improves the one given by Aravind et al.~\cite{matching_cut_structural} for \pname{Matching Cut}, running in time  $\bigOs{12^{\tw(G)}}$.
By employing the cross-composition framework of Bodlaender et al.~\cite{cross_composition} and using a reduction similar to the one in~\cite{matching_cut_ipec}, we show that, unless $\NP \subseteq \coNP/\poly$, there is no polynomial kernel for \pname{$d$-Cut} parameterized simultaneously by the number of crossing edges, the maximum degree, and the treewidth of the input graph.
We then present a polynomial kernel and an $\FPT$ algorithm when parameterizing by the distance to cluster, denoted by $\dc(G)$. This polynomial kernel is our main technical contribution, and it is strongly inspired by the technique presented by  Komusiewicz et al.~\cite{matching_cut_ipec} for \textsc{Matching Cut}. Finally, we give an $\FPT$ algorithm parameterized by the distance to co-cluster, denoted by $\dcc(G)$.
These results imply fixed-parameter tractability for \pname{$d$-Cut} parameterized by $\tau(G)$.

We present in Section~\ref{sec:concl} our concluding remarks and some open questions. We start by providing some basic preliminaries in Section~\ref{sec:prelim}.

\section{Preliminaries}\label{sec:prelim}

We use standard graph-theoretic notation, and we consider simple undirected graphs without loops or multiple edges; see~\cite{Die10} for any undefined terminology. When the graph is clear from the context, the degree (that is, the number of neighbors) of a vertex $v$ is denoted by  $\dgr(v)$, and the number of neighbors of a vertex $v$ in a set $A \subseteq V(G)$ is denoted by $\dgr_A(v)$. The minimum degree, the maximum degree, the line graph, and the vertex cover number of a graph $G$ are denoted by $\delta(G)$, $\Delta(G)$, $L(G)$,
and $\tau(G)$, respectively. For a positive integer $k \geq 1$, we denote by $[k]$ the set containing every integer $i$ such that $1 \leq i \leq k$.

We refer the reader to~\cite{DF13,CyganFKLMPPS15} for basic background on parameterized complexity, and we recall here only some basic definitions.
	A \emph{parameterized problem} is a language $L \subseteq \Sigma^* \times \mathbb{N}$.  For an instance $I=(x,k) \in \Sigma^* \times \mathbb{N}$, $k$ is called the \emph{parameter}. 
	A parameterized problem is \emph{fixed-parameter tractable} ({\sf FPT}) if there exists an algorithm $\Acal$, a computable function $f$, and a constant $c$ such that given an instance $I=(x,k)$,
	$\Acal$ (called an {\sf FPT} \emph{algorithm}) correctly decides whether $I \in L$ in time bounded by $f(k) \cdot |I|^c$.

A fundamental concept in parameterized complexity is that of \emph{kernelization}; see~\cite{book-kernels} for a recent book on the topic. A kernelization
	algorithm, or just \emph{kernel}, for a parameterized problem $\Pi $ takes an
	instance~$(x,k)$ of the problem and, in time polynomial in $|x| + k$, outputs
	an instance~$(x',k')$ such that $|x'|, k' \leqslant g(k)$ for some
	function~$g$, and $(x,k) \in \Pi$ if and only if $(x',k') \in \Pi$. The function~$g$ is called the \emph{size} of the kernel and may
	be viewed as a measure of the ``compressibility'' of a problem using
	polynomial-time preprocessing rules. A kernel is called \emph{polynomial} (resp. \emph{quadratic, linear}) if the function $g(k)$ is a polynomial (resp. quadratic, linear) function in $k$.
	A breakthrough result of Bodlaender et al.~\cite{BodlaenderDFH09} gave the first framework for proving that certain parameterized problems
	do not admit polynomial kernels, by establishing so-called \emph{composition algorithms}. Together with a result of Fortnow and
	Santhanam~\cite{FortnowS11} this allows to exclude polynomial kernels under the assumption that ${\sf NP} \nsubseteq {\sf coNP} / {\sf poly}$, otherwise implying
	a collapse of the polynomial hierarchy to its third level~\cite{Yap83}.

\section{NP-hardness, polynomial cases, and exact exponential algorithm}
\label{sec:np}

In this section we focus on the classical complexity of the \textsc{$d$-Cut} problem, and on exact exponential algorithms. Namely, we provide the \NP-hardness result in Section~\ref{sec:NP-hard}, the polynomial algorithm for graphs of bounded degree in Section~\ref{sec:poly-algo}, and
a simple exact exponential algorithm in Section~\ref{sec:exact-algo}.

\subsection{NP-hardness for regular graphs}
\label{sec:NP-hard}

Before stating our \NP-hardness result, we need some definitions and observations.

\begin{definition}
    A set of vertices $X \subseteq V(G)$ is said to be \emph{monochromatic} if, for any $d$-cut $(A, B)$ of $G$, $X \subseteq A$ or $X \subseteq B$.
\end{definition}

\begin{observation}
    \label{obs:mono_bipartite}
    For fixed $d \geq 1$, the graph $K_{d+1, 2d+1}$ is monochromatic.
    Moreover, any vertex with $d+1$ neighbors in a monochromatic set is monochromatic.
\end{observation}

\begin{definition}[Spool]
    For $n,d \geq 1$, a \emph{$(d, n)$-spool}  is the graph obtained from $n$  copies of $K_{d+1, 2d+2}$ such that, for every $1 \leq i \leq n$, one vertex of degree $d+1$ of the $i$-th copy is identified with one vertex of degree $d+1$ of the $(i+1 \mod n)$-th copy, so that the two chosen vertices in each copy are distinct. 
    The \emph{exterior} vertices of a copy are those of degree $d+1$ that are not used to interface with another copy.
    The \emph{interior} vertices of a copy are those of degree $2d+2$ that do not interface with another copy.
\end{definition}

An illustration of a $(2,3)$-spool is shown in Figure~\ref{fig:spool}.

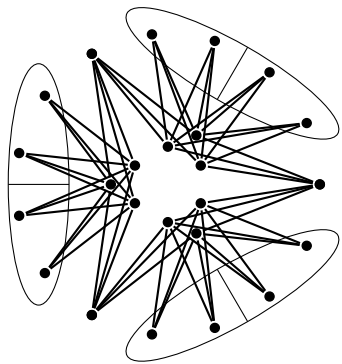
\begin{figure}[!htb]
    \centering
    \begin{tikzpicture}[scale=1]
            \GraphInit[unit=3,vstyle=Simple]
            \SetVertexSimple[Shape=circle, FillColor=black, MinSize=2pt]
            \tikzset{VertexStyle/.append style = {inner sep = \inners, outer sep = \outers}}
            \SetVertexNoLabel
            \foreach \x in {0,1,2} {
                \pgfmathtruncatemacro{\med}{(\x)*120 + 30}
                \begin{scope}[rotate=-\med]
                    \draw (0,1.7) ellipse (1.6cm and 0.4cm);
                    \draw (0,1.3) -- (0,2.1);
                \end{scope}
                \foreach \y in {0,1,2,3,4,5} {
                    \pgfmathtruncatemacro{\ang}{\x * 120 + \y * 24}
                    \Vertex[a=\ang, d=2]{i\x\y}
                }
                \foreach \y in {0,1,2} {
                    \pgfmathtruncatemacro{\ang}{\x * 120 + \y * 30 + 30}
                    \pgfmathsetmacro{\bla}{0.5 + mod(\y,2) *0.25}
                    \Vertex[a=\ang, d=\bla]{o\x\y}
                    \foreach \z in {0,1,2,3,4,5} {
                        \Edge(i\x\z)(o\x\y)
                    }
                }
            }
    \end{tikzpicture}
    \caption{\centering A $(2,3)$-spool. Circled vertices are exterior vertices.\label{fig:spool}}
\end{figure}

\begin{observation}
    For fixed $d\geq 1$, a $(d, n)$-spool is monochromatic.
\end{observation}

\begin{proof}
    Let $S$ be a $(d, n)$-spool.
    If $n=1$, the observation follows by combining the two statements of Observation~\ref{obs:mono_bipartite}.
    Now let $X, Y \subsetneq S$ be two copies of $K_{d+1,2d+2}$ that share exactly one vertex $v$.
    By Observation~\ref{obs:mono_bipartite}, $X' = X \setminus \{v\}$ and $Y' = Y \setminus \{v\}$ are monochromatic.
    Since $v$ has $d+1$ neighbors in $X'$ and $d+1$ in $Y'$, it follows that $X \cup Y$ is monochromatic. By repeating the same argument for every  two copies of $K_{d+1,2d+2}$ that share exactly one vertex, the observation follows.
\end{proof}

Chv\'atal~\cite{chvatal_matching_cut} proved that \textsc{Matching Cut} is \NP-hard for graphs of maximum degree at least four. In the next theorem, whose proof is inspired by the reduction of Chv\'atal~\cite{chvatal_matching_cut}, we prove the \NP-hardness of \pname{$d$-cut} for $(2d+2)$-regular graphs. In particular, for $d=1$ it implies the \NP-hardness of \pname{Matching Cut} for $4$-regular graphs, which is stronger than Chv\'atal~\cite{chvatal_matching_cut} hardness for graphs of maximum degree four.

\begin{theorem}
    \label{thm:regular_nph}
    For every integer $d \geq 1$, \pname{$d$-cut} is \NP-hard even when restricted to $(2d+2)$-regular graphs.
\end{theorem}


\begin{proof}
    Our reduction is from the \pname{3-Uniform Hypergraph Bicoloring} problem, which is \NP-hard; see~\cite{lovasz_hypergraph}. 

    \probl{3-Uniform Hypergraph Bicoloring}{A hypergraph $\mathcal{H}$ with exactly three vertices in each hyperedge.}{Can we 2-color $V(\mathcal{H})$ such that no hyperedge is monochromatic?}

    Throughout this proof, $i$ is an index representing a color, $j$ and $k$ are redundancy indices used to increase the degree of some sets of vertices, and $\ell$ and $r$ are indices used to refer to separations of sets of exterior vertices.

    Given an instance $\mathcal{H}$ of \pname{3-Uniform Hypergraph Bicoloring}, we proceed to construct a $(2d+2)$-regular instance $G$ of \pname{$d$-Cut} as follows. For each vertex $v \in V(\mathcal{H})$, add a $(d,4\dgr(v) + 1)$-spool to $G$.
    Each set of exterior vertices receives an (arbitrarily chosen) unique label from the following types: $S(v^*)$ and $S(v, e, i, j)$, such that $i,j \in [2]$ and $e \in E(\mathcal{H})$ with $v \in e$.
    Separate each of the labeled sets into two parts of equal size (see Figure~\ref{fig:spool}).
    For the first type, we denote the sets by $S(v^*, i)$, $i \in [2]$; for the second type, by $S_{\ell}(v, e, i, j)$, $\ell \in [2]$.
    For each set $S(v^*, i)$, we choose an arbitrary vertex and label it with $s(v^*, i)$.
    To conclude the construction of vertex gadgets, add every edge between $S_1(v, e, i, j)$ and $S_2(v, e, i, j)$, and form a perfect matching between $S(v^*, 1) \setminus \{s(v^*, 1)\}$ and $S(v^*, 2) \setminus \{s(v^*, 2)\}$.
    Note that all inner vertices of these spools have degree $2d+2$, every vertex labeled $s(v^*, i)$ has $d+1$ neighbors, every other vertex in $S(v^*, i)$ has $d+2$, and every vertex in $S(v, e, i, j)$ has degree equal to $2d+1$.
    For an example of the edges between exterior vertices of the same vertex gadget, see Figure~\ref{fig:vert_relations}.

    \begin{figure}[!htb]
        \centering
        \hspace*{-0.2cm}\begin{tikzpicture}[scale=1]
            \begin{scope}[shift={(0cm,0)}]
                \GraphInit[unit=3,vstyle=Normal]
                \SetVertexNormal[Shape=circle, FillColor=black, MinSize=2pt]
                \tikzset{VertexStyle/.append style = {inner sep = \inners, outer sep = \outers}}
                \pgfmathtruncatemacro{\zero}{0}
                \pgfmathtruncatemacro{\one}{1}
                \pgfmathtruncatemacro{\six}{6}
                    \foreach \w in {1} {
                        \foreach \x in {0,1} {
                            \foreach \y in {1,2,3,4,5,6} {
                                \pgfmathsetmacro{\bla}{(6.0 + 12.0/7.0) * \x + 5.0/7.0 * \y}
                                \pgfmathsetmacro{\bly}{-3 * \w}
                                \pgfmathtruncatemacro{\wo}{\w + 1}
                                \ifnum\x=\zero
                                    \Vertex[x=\bla, y=\bly, NoLabel]{e\w\x\y}
                                \else
                                    \ifnum\y=\one
                                        \Vertex[x=\bla, y=\bly, Math, LabelOut, Ldist=-4pt,Lpos=135,L={s(v^*, 1)}]{e\w\x\y}
                                    \else
                                        \ifnum\y=\six
                                          \Vertex[x=\bla, y=\bly, Math, LabelOut, Ldist=1pt,Lpos=0,L={s(v^*, 2)}]{e\w\x\y}
                                        \else
                                            \Vertex[x=\bla, y=\bly, NoLabel]{e\w\x\y}
                                        \fi
                                    \fi
                                \fi
                            }
                            \pgfmathsetmacro{\bla}{2.5 + (6.0 + 12.0/7.0) * \x}
                            \pgfmathsetmacro{\blam}{\bla - 1.2}
                            \pgfmathsetmacro{\blap}{\bla + 1.2}
                            \ifnum\w=\zero
                                \draw (\bla,-0.7) -- (\bla,1.2);
                            \else
                                \draw (\bla,-4.2) -- (\bla,-2.3);
                            \fi
                            \ifnum\w=\zero
                                \ifnum\x=\zero
                                    \node at (\blam,1) {$S_1(v, e, i, 1)$};
                                    \node at (\blap,1) {$S_2(v, e, i, 1)$};
                                \else
                                    \node at (\blam,1) {$S_1(v^*, 1)$};
                                    \node at (\blap,1) {$S_2(v^*, 1)$};
                                \fi
                            \else
                                \ifnum\x=\zero
                                    \node at (\blam,-4) {$S_1(v, e, i, 2)$};
                                    \node at (\blap,-4) {$S_2(v, e, i, 2)$};
                                \else
                                    \node at (\blam,-4) {$S(v^*, 1)$};
                                    \node at (\blap,-4) {$S(v^*, 2)$};
                                \fi
                            \fi
                        }
                    \foreach \z in {4,5,6} {
                        \Edge[style = bend left](e\w01)(e\w0\z)
                        \Edge[style = bend right](e\w02)(e\w0\z)
                    }
                    \Edge(e\w03)(e\w04)
                    \Edge[style = bend left](e\w03)(e\w05)
                    \Edge[style = bend left](e\w03)(e\w06)
                    \Edge[style = bend left](e\w12)(e\w15)
                    \Edge(e\w13)(e\w14)
                }
            \end{scope}
        \end{tikzpicture}
        \caption{\centering Relationships between exterior vertices of a vertex gadget ($d=3$).\label{fig:vert_relations}}
    \end{figure}
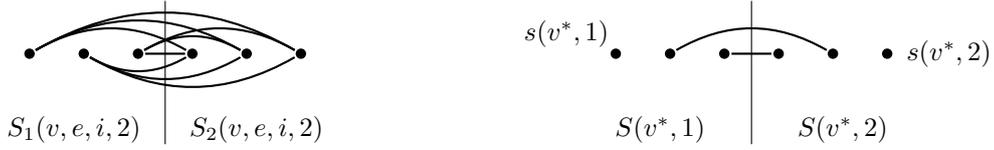

    For each color $i \in [2]$, add a $(d, n + 2m)$-spool to $G$, where $n = |V(\mathcal{H})|$ and $m = |E(\mathcal{H})|$.
    Much like the exterior vertices of the vertex gadgets, we attribute unique labels: $C(v, i)$, for each $v \in V(\mathcal{H})$, and $C(e, i, j)$, for each $e \in E(\mathcal{H})$ and $j \in [2]$.
    Now, split the remaining vertices of each labeled set into two equal-sized parts $C_1(\cdot), C_2(\cdot)$ and label one vertex of each $C_{\ell}(e, i, j)$ with the label $c_{\ell}(e, i, j)$ and one of each $C_{\ell}(v, i)$ with $c_{\ell}(v,i)$.
    To conclude, add all edges from $c_{\ell}(v, i)$ to $C_{\ell}(v, i)$, add the edge $c_{\ell}(v, i)c_{3-\ell}(v, i)$, make each $C_{\ell}(e, i, j) \setminus \{c_{\ell}(e, i, j)\}$ into a clique, and, between $C_{\ell}(e, i, j) \setminus \{c_{\ell}(e, i, j)\}$ and $C_r(e, i, k) \setminus \{c_r(e, i, k)\}$, add edges to form a perfect matching, for $\ell,j,r,k \in [2]$.
    That is, each $C_{\ell}(e, i, j)$ forms a perfect matching with three other sets of exterior vertices.
    So far, each $c_{\ell}(v, i)$ has degree $(d+1) + (d-1) + 1 = 2d+1$, other vertices of $C_{\ell}(v, i)$ have degree $d+2$, each vertex in $C_{\ell}(e, i, j) \setminus \{c_{\ell}(e, i, j)\}$ has degree $(d + 1) + (d - 2) + 3 = 2d + 2$, and each vertex labeled $c_{\ell}(e, i, j)$ has degree $d + 1$.

    We now add edges between vertices of different color gadgets.
    In particular, we add every edge between $C_1(v, 2) \setminus \{c_1(v, 2)\}$ and $C_2(v, 1) \setminus \{c_2(v, 1)\}$.
    This increases the degree of these vertices to $2d+1$.
    An example when $d=3$ is illustrated in Figure~\ref{fig:color_relations}.

    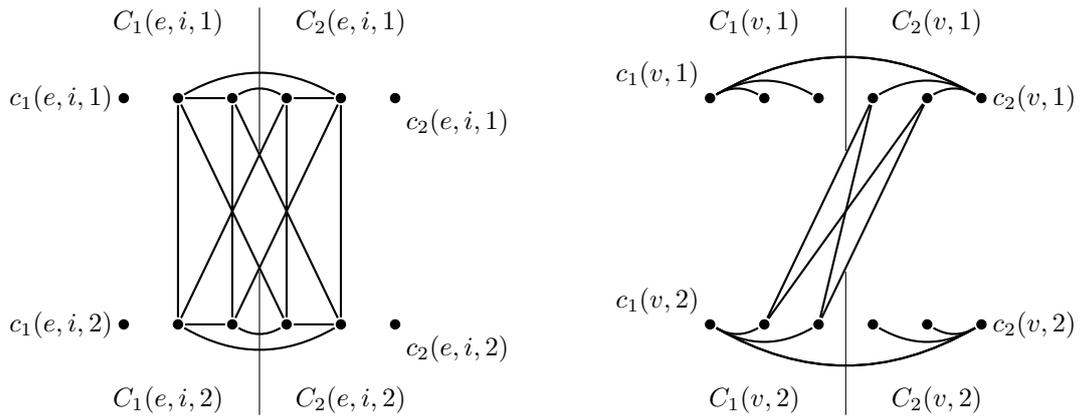
\begin{figure}[!htb]
        \centering
        \hspace*{-1cm}\begin{tikzpicture}[scale=1]
            \begin{scope}[shift={(0cm,0)}]
                \GraphInit[unit=3,vstyle=Normal]
                \SetVertexNormal[Shape=circle, FillColor=black, MinSize=2pt]
                \tikzset{VertexStyle/.append style = {inner sep = \inners, outer sep = \outers}}
                \pgfmathtruncatemacro{\zero}{0}
                \pgfmathtruncatemacro{\one}{1}
                \pgfmathtruncatemacro{\six}{6}
                    \foreach \w in {0, 1} {
                        \foreach \x in {0,1} {
                            \foreach \y in {1,2,3,4,5,6} {
                                \pgfmathsetmacro{\bla}{(6.0 + 12.0/7.0) * \x + 5.0/7.0 * \y}
                                \pgfmathsetmacro{\bly}{-3 * \w}
                                \pgfmathtruncatemacro{\wo}{\w + 1}
                                \ifnum\x=\zero
                                   \ifnum\y=\one
                                    \Vertex[x=\bla, y=\bly, Math, LabelOut, Lpos=180, Ldist=-2pt, L={c_1(e,i,\wo)}]{e\w\x\y}
                                   \else
                                       \ifnum\y=\six
                                           \Vertex[x=\bla, y=\bly, Math, LabelOut, Lpos=-45, Ldist=-2pt, L={c_2(e,i,\wo)}]{e\w\x\y}
                                       \else
                                        \Vertex[x=\bla, y=\bly, NoLabel]{e\w\x\y}
                                       \fi
                                    \fi
                                \else
                                    \ifnum\w=\one
                                       \ifnum\y=\one
                                        \Vertex[x=\bla, y=\bly, Math, LabelOut, Lpos=135, Ldist=-2pt, L={c_1(v,2)}]{e\w\x\y}
                                       \else
                                           \ifnum\y=\six
                                               \Vertex[x=\bla, y=\bly, Math, LabelOut, Lpos=0, Ldist=-2pt, L={c_{2}(v, 2)}]{e\w\x\y}
                                           \else
                                            \Vertex[x=\bla, y=\bly, NoLabel]{e\w\x\y}
                                           \fi
                                        \fi
                                    \else
                                        \ifnum\y=\one
                                        \Vertex[x=\bla, y=\bly, Math, LabelOut, Lpos=135, Ldist=-2pt, L={c_1(v,1)}]{e\w\x\y}
                                       \else
                                           \ifnum\y=\six
                                               \Vertex[x=\bla, y=\bly, Math, LabelOut, Lpos=0, Ldist=-2pt, L={c_{2}(v, 1)}]{e\w\x\y}
                                           \else
                                            \Vertex[x=\bla, y=\bly, NoLabel]{e\w\x\y}
                                           \fi
                                        \fi
                                    \fi
                                \fi
                            }
                            \ifnum\x=\zero
                                \Edge(e\w\x2)(e\w\x3)
                                \Edge(e\w\x4)(e\w\x5)
                                \ifnum\w=\zero
                                    \Edge[style=bend left](e\w\x2)(e\w\x5)
                                    \Edge[style=bend left](e\w\x3)(e\w\x4)
                                \else
                                    \Edge[style=bend right](e\w\x2)(e\w\x5)
                                    \Edge[style=bend right](e\w\x3)(e\w\x4)
                                \fi
                            \else
                                \newcommand{\bla}{right}
                                \ifnum\w=\zero
                                    \renewcommand{\bla}{left}
                                \fi

                                \foreach \y in {2,3} {
                                    \pgfmathtruncatemacro{\z}{7 - \y}
                                    \Edge[style=bend \bla](e\w\x1)(e\w\x\y)
                                    \Edge[style=bend \bla](e\w\x\z)(e\w\x6)
                                    \Edge[style=bend \bla](e\w\x1)(e\w\x6)
                                }
                            \fi
                            \pgfmathsetmacro{\bla}{2.5 + (6.0 + 12.0/7.0) * \x}
                            \pgfmathsetmacro{\blam}{\bla - 1.2}
                            \pgfmathsetmacro{\blap}{\bla + 1.2}
                            \ifnum\w=\zero
                                \draw (\bla,-0.7) -- (\bla,1.2);
                            \else
                                \draw (\bla,-4.2) -- (\bla,-2.3);
                            \fi
                            \ifnum\w=\zero
                                \ifnum\x=\zero
                                    \node at (\blam,1) {$C_1(e, i, 1)$};
                                    \node at (\blap,1) {$C_2(e, i, 1)$};
                                \else
                                    \node at (\blam,1) {$C_1(v, 1)$};
                                    \node at (\blap,1) {$C_2(v, 1)$};
                                \fi
                            \else
                                \ifnum\x=\zero
                                    \node at (\blam,-4) {$C_1(e, i, 2)$};
                                    \node at (\blap,-4) {$C_2(e, i, 2)$};
                                \else
                                    \node at (\blam,-4) {$C_1(v, 2)$};
                                    \node at (\blap,-4) {$C_2(v, 2)$};
                                \fi
                            \fi
                        }
                }
                \foreach \y in {2,3,4,5} {
                    \Edge(e00\y)(e10\y)
                }
                \foreach \y in {2,3} {
                    \pgfmathtruncatemacro{\yp}{\y+2}
                    \pgfmathtruncatemacro{\ym}{7 - \y}
                    \Edge(e00\y)(e10\yp)
                    \Edge(e10\y)(e00\yp)
                    \foreach \z in {4,5} {
                        \Edge(e11\y)(e01\z)
                    }
                }
            \end{scope}
        \end{tikzpicture}
        \caption{\centering Relationships between exterior vertices of color gadgets ($d=3$).\label{fig:color_relations}}
    \end{figure}

    As a first step to connect color gadgets and vertex gadgets, we add every edge between $s(v^*, i)$ and $C_i(v, i)$, every edge between $S(v^*, i) \setminus \{s(v^*, i)\}$ and $C_i(v, i) \setminus \{c_{i}(v, i)\}$, a perfect matching between $S(v^*, i) \setminus \{s(v^*, i)\}$ and $C_{3-i}(v, i) \setminus \{c_{3-i}(v, i)\}$, and the edge $s(v^*, i)c_i(v, 3-i)$.
    Note that this last edge is fundamental, not only because it increases the degrees to the desired value, but also because, if both color gadgets belong to the same side of the cut, every $s(v^*, i)$ will have the same color and, since spools are monochromatic, so would be the {\sl entire} graph, as discussed in more detail below.
    Also note that, aside from $s(v^*, i)$, no other vertex has more than $d$ neighbors outside of its spool.
    The edges described in this paragraph increase the degree of every $s(v^*, i)$ by $d + 1$, yielding a total degree of $2d+2$, of every vertex in $S(v^*, i) \setminus \{s(v^*, i)\}$ to $(d+2) + (d-1) + 1 = 2d+2$, of every vertex in $C_{i}(v, i) \setminus \{c_i(v, i)\}$ to $(d+2) + d = 2d+2$, of every vertex in $C_{i}(v, 3-i) \setminus \{c_i(v, 3-i)\}$ to $(2d+1) + 1 = 2d+2$, and of every $c_{\ell}(v, i)$ to $(2d+1) + 1 = 2d+2$.
    Figure~\ref{fig:color_vertex_relations} gives an example of these connections.

    \begin{figure}[!htb]
        \centering
        \hspace*{-1cm}\begin{tikzpicture}[scale=1]
            \begin{scope}[shift={(0cm,0)}]
                \GraphInit[unit=3,vstyle=Normal]
                \SetVertexNormal[Shape=circle, FillColor=black, MinSize=2pt]
                \tikzset{VertexStyle/.append style = {inner sep = \inners, outer sep = \outers}}
                \pgfmathtruncatemacro{\zero}{0}
                \pgfmathtruncatemacro{\one}{1}
                \pgfmathtruncatemacro{\six}{6}
                    \foreach \w in {0, 1} {
                        \foreach \x in {1} {
                            \foreach \y in {1,2,3,4,5,6} {
                                \pgfmathsetmacro{\bla}{(6.0 + 12.0/7.0) * \x + 5.0/7.0 * \y}
                                \pgfmathsetmacro{\bly}{-3 * \w}
                                \pgfmathtruncatemacro{\wo}{\w + 1}
                                \ifnum\x=\zero
                                \else
                                    \ifnum\w=\one
                                       \ifnum\y=\one
                                        \Vertex[x=\bla, y=\bly, Math, LabelOut, Lpos=135, Ldist=-2pt, L={s(v^*,1)}]{e\w\x\y}
                                       \else
                                           \ifnum\y=\six
                                               \Vertex[x=\bla, y=\bly, Math, LabelOut, Lpos=0, Ldist=-2pt, L={c_{1}(v, 2)}]{e\w\x\y}
                                           \else
                                            \Vertex[x=\bla, y=\bly, NoLabel]{e\w\x\y}
                                           \fi
                                        \fi
                                    \else
                                        \ifnum\y=\one
                                        \Vertex[x=\bla, y=\bly, Math, LabelOut, Lpos=135, Ldist=-2pt, L={c_1(v,1)}]{e\w\x\y}
                                       \else
                                           \ifnum\y=\six
                                               \Vertex[x=\bla, y=\bly, Math, LabelOut, Lpos=0, Ldist=-2pt, L={c_{2}(v, 1)}]{e\w\x\y}
                                           \else
                                            \Vertex[x=\bla, y=\bly, NoLabel]{e\w\x\y}
                                           \fi
                                        \fi
                                    \fi
                                \fi
                            }
                            \ifnum\x=\zero
                                \Edge(e\w\x2)(e\w\x3)
                                \Edge(e\w\x4)(e\w\x5)
                                \ifnum\w=\zero
                                    \Edge[style=bend left](e\w\x2)(e\w\x5)
                                    \Edge[style=bend left](e\w\x3)(e\w\x4)
                                \else
                                    \Edge[style=bend right](e\w\x2)(e\w\x5)
                                    \Edge[style=bend right](e\w\x3)(e\w\x4)
                                \fi
                            \else
                                \newcommand{\bla}{right}
                                \ifnum\w=\zero
                                    \renewcommand{\bla}{left}
                                \fi
                                 \foreach \y in {2,3} {
                                     \pgfmathtruncatemacro{\z}{7 - \y}
                                     \ifnum\w=\zero
                                         \Edge[style=bend \bla](e\w\x1)(e\w\x\y)
                                     \fi
                                    \Edge[style=bend \bla](e\w\x1)(e\w\x6)
                                     \Edge[style=bend \bla](e\w\x\z)(e\w\x6)
                                 }
                            \fi
                            \pgfmathsetmacro{\bla}{2.5 + (6.0 + 12.0/7.0) * \x}
                            \pgfmathsetmacro{\blam}{\bla - 1.2}
                            \pgfmathsetmacro{\blap}{\bla + 1.2}
                            \ifnum\w=\zero
                                \draw (\bla,-0.7) -- (\bla,1.2);
                            \else
                                \draw (\bla,-4.2) -- (\bla,-2.3);
                            \fi
                            \ifnum\w=\zero
                                \ifnum\x=\zero
                                \else
                                    \node at (\blam,1) {$C_1(v, 1)$};
                                    \node at (\blap,1) {$C_2(v, 1)$};
                                \fi
                            \else
                                \ifnum\x=\zero
                                \else
                                    \node at (\blam,-4) {$S(v^*, 1)$};
                                    \node at (\blap,-4) {$C_1(v, 2)$};
                                \fi
                            \fi
                        }
                }
                \Edge(e111)(e011)
                \foreach \y in {1,2,3} {
                    \foreach \z in {2,3} {
                        \Edge(e11\y)(e01\z)
                    }
                }
                \foreach \y in {4,5} {
                    \pgfmathtruncatemacro{\ym}{\y - 2}
                    \Edge(e11\ym)(e01\y)
                    \foreach \z in {4,5} {
                        \Edge(e11\y)(e01\z)
                    }
                }
            \end{scope}
        \end{tikzpicture}
        \caption{\centering Relationships between exterior vertices of color and vertex gadgets ($d=3$).\label{fig:color_vertex_relations}}
    \end{figure}
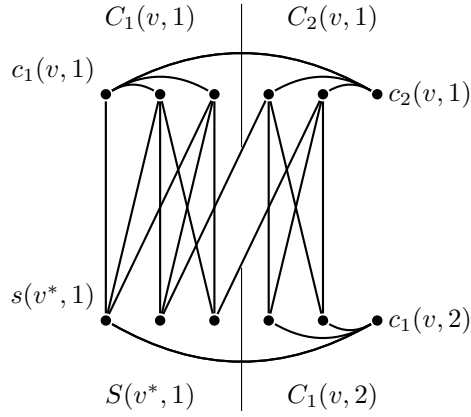

    For the final group of gadgets, namely hyperedge gadgets, for each $\{x, y, z\} \in E(\mathcal{H})$, each color $i$, and each pair $j,\ell \in [2]$, we add one additional vertex $c_{\ell}'(e, i, j)$ adjacent to $c_{\ell}(e, i, j)$, $S_{\ell}(x, e, i, j)$, $S_{\ell}(y, e, i , j)$, and $c_{3-\ell}'(e, i, j)$; finally, we add every edge between $c_{\ell}(e, i, j)$ and $S_{\ell}(z, e, i ,j)$. See Figure~\ref{fig:hyper_relations} for an illustration.
    Note that $c_{\ell}'(e, i, j)$ has degree $2d + 2$; the degree of $c_{\ell}(e, i ,j)$ increased from $d + 1$ to $2d + 2$, and the degree of each vertex of $S_{\ell}(x, e, i, j)$ increased from $2d + 1$ to $2d + 2$.
    This concludes our construction of the $(2d + 2)$-regular graph $G$.

    \begin{figure}[!htb]
        \centering
        \hspace*{-0.25cm}\begin{tikzpicture}[scale=1]
            \begin{scope}[shift={(0cm,0)}]
                \GraphInit[unit=3,vstyle=Normal]
                \SetVertexNormal[Shape=circle, FillColor=black, MinSize=2pt]
                \tikzset{VertexStyle/.append style = {inner sep = \inners, outer sep = \outers}}
                \pgfmathtruncatemacro{\zero}{0}
                \pgfmathtruncatemacro{\one}{1}
                \pgfmathtruncatemacro{\six}{6}
                \foreach \x in {0,1} {
                    \pgfmathtruncatemacro{\sh}{(6.29 + 12.0/7.0) * \x}
                    \begin{scope}[shift={(\sh,0)}]
                        \pgfmathtruncatemacro{\xp}{1 + 1*\x}
                        \pgfmathtruncatemacro{\pe}{180*(\x)}
                        \pgfmathtruncatemacro{\pep}{45 + 90*(1-\x)}
                        \Vertex[x=0, y=1.1, Math, LabelOut, Ldist=-2pt,Lpos=\pe,L={c_\xp(e,i,j)}]{e\x}
                        \Vertex[x=0, y=0.2, Math, LabelOut, Ldist=-2pt,Lpos=\pep,L={c'_\xp(e,i,j)}]{ep\x}
                        \Edge(e\x)(ep\x)
                        \foreach \w in {0,1,2} {
                            \pgfmathsetmacro{\ang}{120*(\w+1)}
                            \begin{scope}[rotate=\ang]
                                \foreach \y in {1,2,3} {
                                    \pgfmathsetmacro{\gap}{5.0/8.0}
                                    \pgfmathsetmacro{\bla}{\gap * \y - 2*\gap}
                                    \pgfmathtruncatemacro{\ym}{\y - 1}
                                    \newcommand{\setname}{z}
                                    \newcommand{\lpos}{270}
                                    \newcommand{\ldist}{13pt}
                                    \ifnum\w>\one
                                        \renewcommand{\lpos}{90}
                                        \renewcommand{\ldist}{0pt}
                                    \fi
                                    \ifnum\w=\zero
                                        \renewcommand{\setname}{x}
                                    \fi
                                    \ifnum\w=\one
                                        \renewcommand{\setname}{y}
                                    \fi
                                    \ifnum\ym=\one
                                        \Vertex[x=\bla, y=1.8, Math, LabelOut, Lpos=\lpos, Ldist=\ldist, L={S_{\xp}(\setname, e, i, j)}]{e\w\x\y}
                                    \else
                                        \Vertex[x=\bla, y=1.8, NoLabel]{e\w\x\y}
                                    \fi
                                    \ifnum\w>\one
                                        \Edge(e\x)(e\w\x\y)
                                    \else
                                        \Edge(ep\x)(e\w\x\y)
                                    \fi
                                }
                            \end{scope}
                        }
                    \end{scope}
                }
                \Edge(ep0)(ep1)
            \end{scope}
        \end{tikzpicture}
        \caption{\centering Hyperedge gadget ($d=3$).\label{fig:hyper_relations}}
    \end{figure}
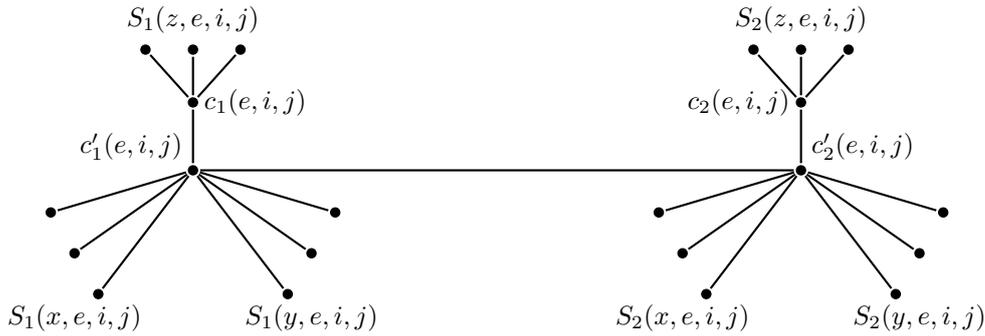

    Now, suppose we are given a valid bicoloring $\varphi$ of $\mathcal{H}$, and our goal is to construct a $d$-cut $(A, B)$ of $G$. Put the gadget of color 1 in $A$ and the other one in $B$.
    For each vertex $v \in V(\mathcal{H})$, if $\varphi(v) = 1$, put the gadget corresponding to $v$ in $A$, otherwise put it in $B$.
    In the interface between these gadgets, no vertex from the color gadgets has more than $d$ neighbors in a single vertex gadget, therefore none violates the $d$-cut property.
    As to the vertices coming from the vertex gadgets, only $s_{\ell}(v^*, i)$ has more than $d$ neighbors outside of its gadget; however, it has $d$ neighbors in the color gadget for color $i$ and only one in color $3-i$.
    Since each color gadget is in a different side of the partition, $s_{\ell}(v^*, i)$ does not violate the degree constraint.
    For each hyperedge $e = \{x, y, z\}$, put $c_{\ell}'(e, i, j)$ in the same set as the majority of its neighbors, this way, it will not violate the property -- note that its other neighbor, $c_{3-\ell}'(e, i , j)$, will be in the same set because it will have the exact same amount of vertices on each side of the partition in its neighborhood.
    So, if $\varphi(x) = \varphi(y) = 1$, $c_{\ell}'(e, i, j) \in A$; however, since $e$ is not monochromatic, $\varphi(z) = 2$, so $c_{\ell}(e, i ,j)$ has at most $d$ neighbors in the other set.
    The case where $\varphi(x) \neq \varphi(y)$ is similar.
    Thus, we conclude that $(A, B)$ is indeed a $d$-cut of $G$.

    Conversely, take a $d$-cut $(A, B)$ of $G$ and construct a bicoloring of $\mathcal{H}$ such that $\varphi(v) = 1$ if and only if the spool corresponding to $v$ is in $A$.
    Suppose that this process results in some hyperedge $e = \{x, y, z\} \in E(\mathcal{H})$ begin monochromatic.
    That is, there is some hyperedge gadget where $S_{\ell}(x, e, i, j)$, $S_{\ell}(y, e, i, j)$, and $S_{\ell}(z, e, i, j)$ are in $A$, which implies that $c_{\ell}'(e, i, j) \in A$ and, consequently, that $c_{\ell}(e, i, j) \in A$ for every $\ell,i,j \in [2]$.
    However, since $c_{\ell}(e, 1, j)$ and $c_{\ell}(e, 2, j)$ are in $A$ and a color gadget is monochromatic, both color gadgets belong to $A$, which in turn implies that every $s(v^*, i)$ has $d+1$ neighbors in $A$ and, therefore, must also be in $A$ by Observation~\ref{obs:mono_bipartite}.
    Moreover, since spools are monochromatic, every vertex gadget is in $A$, implying that the entire graph belongs to $A$, contradicting the hypothesis that $(A, B)$ is a $d$-cut of $G$.
\end{proof}

The graphs constructed by the above reduction are neither planar nor bipartite, but they are regular, a result that we were unable to find in the literature for \pname{Matching Cut}.
Note that every planar graph has a $d$-cut for every $d \geq 5$, so only the cases $d \in \{2,3,4\}$ remain open, as the case $d=1$ is known to be $\NP$-hard~\cite{matching_cut_planar}. 
Concerning graphs of bounded diameter, Le and Le~\cite{matching_cut_diameter} prove the $\NP$-hardness of \textsc{Matching Cut} for graphs of diameter at least three by reducing \textsc{Matching Cut} to itself.  It can be easily seen that the same construction given by Le and Le~\cite{matching_cut_diameter}, but reducing  \pname{$d$-Cut} to itself, also proves the $\NP$-hardness of \pname{$d$-Cut} for every $d \geq 1$. 


\begin{corollary}
    For every integer $d \geq 1$, \pname{$d$-Cut} is $\NP$-hard for graphs of diameter at least three.
\end{corollary}

We leave as an open problem to determine whether there exists a  polynomial-time algorithm for \pname{$d$-Cut} for graphs of diameter at most two for every $d \geq 2$, as it is the case for $d=1$~\cite{matching_cut_diameter}.


\subsection{Polynomial algorithm for graphs of bounded degree}
\label{sec:poly-algo}

Our next result is a natural generalization of Chvátal's algorithm~\cite{chvatal_matching_cut} for \pname{Matching Cut} on graphs of maximum degree three.

\begin{theorem}
    \label{thm:small_deg_poly}
    For any graph $G$ and integer $d \geq 1$ such that $\Delta(G) \leq d+2$, it can be decided in polynomial time if $G$ has a $d$-cut. Moreover, for $d=1$ any graph $G$ with $\Delta(G) \leq 3$ and $|V(G)| \geq 8$ has a matching cut, for $d=2$ any graph $G$ with $\Delta(G) \leq 4$ and $|V(G)| \geq 6$ has a $2$-cut, and for $d\geq 3$ any graph $G$ with $\Delta(G) \leq d+2$ has a $d$-cut.
\end{theorem}

\begin{proof}
    We may assume that $G$ is connected, as otherwise it always admits a $d$-cut. If $G$ is a tree, any edge is a cut edge and, consequently, a $d$-cut is easily found.
    So let $C$ be a shortest cycle of $G$.
    If $d = 1$ we use Chvátal's result~\cite{chvatal_matching_cut} together with the size bound of eight observed by Moshi~\cite{matching_cut_moshi}; hence, we may assume that $d \geq 2$.
    In the case that $V(G) = C$, we may pick any vertex $v$ and note that $(\{v\}, C \setminus \{v\})$ is a $d$-cut.


     Suppose first that $|C| = 3$ and $d = 2$.  If $(C, V(G) \setminus C)$ is a $2$-cut, we are done. Otherwise, there is some vertex $v \notin C$ with three neighbors in $C$ (since by the hypothesis on $\Delta(G)$, every vertex in $C$ has at most two neighbors in $G - C$) and, consequently, $Q := C \cup \{v\}$ induces a $K_4$.
    If $V(G) = Q$, we can arbitrarily partition $Q$ into two sets with two vertices each and get a $2$-cut of $G$.
    Also, if no other $u \notin Q$ has three neighbors in $Q$, $(Q, V(G) \setminus Q)$ is a $2$-cut of $G$.
    If there is such a vertex $u$, let $R := Q \cup \{u\}$. If $V(G) = R$, then clearly $G$ has no $2$-cut. Note that $|Q|=5$, and this will be the only case in the proof where $G$ does not have a $d$-cut. Otherwise, if $V(G) \neq R$, $(R, V(G) \setminus R)$ is a $2$-cut, because no vertex outside of $R$ can be adjacent to more than two vertices in $R$, and we are done.

    If $|C| = 3$ and $d \geq 3$, then clearly $(C, V(G) \setminus C)$ is a $d$-cut, and we are also done.

    Otherwise, that is, if $|C| \geq 4$, we claim that $(C, V(G) \setminus C)$ is always a $d$-cut.
    For $v \in C$, note that $\dgr(v) \leq d + 2$, hence $v$ has at most $d$ neighbors in $G - C$. For $v \in V(G) \setminus C$, if $|C| \geq 5$, necessarily $\dgr_C(v) \leq 1$, as otherwise we would find a cycle in $G$ shorter than $C$, and therefore $(C, V(G) \setminus C)$ is a $d$-cut.
    By a similar argument, if $|C| = 4$, then $\dgr_C(v) \leq 2$, and the theorem follows as we assume that $d \geq 2$.
\end{proof}

Theorems~\ref{thm:regular_nph} and~\ref{thm:small_deg_poly} present a ``quasi-dichotomy'' for $d$-cut on graphs of bounded maximum degree.
Specifically, for $\Delta(G) \in \{d+3, \dots, 2d+1\}$, the complexity of the problem remains unknown.
However, we believe that most, if not all, of these open cases can be solved in polynomial time; see the discussion in Section~\ref{sec:concl}.

\subsection{Exact exponential algorithm}
\label{sec:exact-algo}

To conclude this section, we present a simple exact exponential algorithm which, for every $d \geq 1$, runs in time $\bigOs{c_d^n}$ for some constant $c_d < 2$.
For the case $d=1$, the currently known algorithms~\cite{matching_cut_tcs,matching_cut_ipec} exploit structures that appear to get out of control when $d$ increases, and so has a better running time than the one described below.

When an instance of size $n$ branches into $t$ subproblems of sizes at most $n - s_1, \dots, n - s_t$, respectively, the vector $(s_1, \dots, s_t)$ is called the \tdef{branching vector} of this branching rule, and the unique positive real root of the equation $x^n - \sum_{i \in [t]} x^{n - s_i} = 0$ is called the \tdef{branching factor} of the rule.
The total complexity of a branching algorithm is given by $\bigOs{\alpha^n}$, where $\alpha$ is the largest branching factor among all rules of the algorithm.
For more on branching algorithms, we refer to~\cite{exact_exponential_algorithms}.

\begin{theorem}
    For every fixed integer $d \geq 1$ and $n$-vertex graph $G$, there is an algorithm that solves \pname{$d$-Cut} in time $\bigOs{c_d^n}$, for some constant $1 < c_d < 2$.
\end{theorem}

\begin{proof}
    Our algorithm takes as input $G$ and outputs a $d$-cut $(A, B)$ of $G$, if it exists.
    To do so, we build a branching algorithm that maintains, at every step, a tripartition of $V(G) = A \dot\cup B \dot\cup D$ such that $(A, B)$ is a $d$-cut of $G \setminus D$.
    The central idea of our rules is to branch on small  sets of vertices (namely, of size at most $d+1$) at each step such that either at least one bipartition of the set forces some other vertex to choose a side of the cut, or we can conclude that there is at least one bipartition that violates the $d$-cut property.
    First, we present our reduction rules, which are applied following this order at the beginning of each recursive step.

    \begin{itemize}
        \item[R1] If $(A, B)$ violates the $d$-cut property, output $\NOi$.
        \item[R2] If $D = \emptyset$, we have a $d$-cut of $G$. Output $(A,B)$.
        \item[R3] If there is some $v \in D$ with $\dgr_A(v) \geq d + 1$ and $\dgr_B(v) \geq d+1$, output $\NOi$.
        \item[R4] While there is some $v \in D$ with $\dgr_A(v) \geq d + 1$ (resp. $\dgr_B(v) \geq d+1$), add $v$ to $A$ (resp. $B$).
    \end{itemize}

    Our branching rules, and their respective branching vectors, are listed below.

    \begin{itemize}
        \item[B1] If there is some $v \in A \cup B$ with $\dgr_D(v) \geq d+1$, choose a set $X \subseteq N_D(v)$ of size $d$ and branch on all possible possible bipartitions of $X$.
        Note that, if all vertices of $X$ are in the other side of $v$, at least one vertex of $N_D(v) \setminus X$ must be in the same side as $v$.
        As such, this branching vector is of the form $\{d+1\} \times \{d\}^{2^d-1}$.

        %

        \item[B2] If there is some $v \in A$ (resp. $B$) such that $\dgr_B(v) + \dgr_D(v) \geq d+1$ (resp. $\dgr_A(v) + \dgr_D(v) \geq d+1$), choose a set $X \subseteq N_D(v)$ of size $s = d+1 - \dgr_B(v)$ (resp. $s = d+1 - \dgr_A(v)$) and branch on every possible bipartition of $X$.
        Since rule B1 was not applied, we have that $\dgr_D(v) \leq d$, $\dgr_B(v) \geq 1$ (resp. $\dgr_A(v) \geq 1$), and $s \leq d$.
        If all vertices of $X$ were placed in $B$ (resp. $A$), we would violate the $d$-cut property and, thus, do not need to investigate this branch of the search.
        In the worst case, namely when $s = d$, this yields the branching vector $\{d\}^{2^d-1}$.
    \end{itemize}

    We now claim that, if none of the above rules is applicable, we have that $(A \cup D, B)$ is a $d$-cut of $G$.
    To see that this is the case, suppose that there is some vertex $v \in V(G)$ that violates the $d$-cut property; that is, it has a set $Y$ of $d+1$ neighbors across the cut.

    Suppose that $v \in B$. Then $Y \subseteq A \cup D$, so we have $\dgr_A(v) + \dgr_D(v) \geq d+1$, in which case rule B2 could be applied, a contradiction.
    Thus, we have that $v \notin B$, so $Y \subseteq B$ and either $v \in A$ or $v \in D$;
    in the former case, again by rule R1, $(A, B)$ would not be a $d$-cut.
     In the latter case, we would have that $\dgr_B(v) \geq d+1$, but then rule R4 would still be applicable.
    Consequently, $v \notin A \cup B \cup D = V(G)$, so such a vertex does not exist, and thus we have that $(A \cup D, B)$ is a $d$-cut of $G$.
    Note that a symmetric argument holds for the bipartition $(A, B \cup D)$.
    Before executing the above branching algorithm, we need to ensure that $A \neq \emptyset$ and $B \neq \emptyset$.
    To do that, for each possible pair of vertices $u,v \in V(G)$, we execute the entire algorithm starting with $A := \{u\}$ and $B := \{v\}$.

    As to the running time of the algorithm, for rule B2 we have that the unique positive real root of $x^n - (2^d - 1)x^{n - d} = 0$ is of the closed form $x = \sqrt[\leftroot{2}d]{2^d - 1} < 2$.
    For rule B1, we have that the polynomial associated with the recurrence relation, $p_d(x) = x^n - (2^d - 1)x^{n - d} - x^{n - d - 1}$, verifies $p_d(1) = 1 - 2^d < 0$ and $p_d(2) = 2^{n - d - 1} > 0$.
    Since it is a continuous function and $p_d(x)$ has an unique positive real root $c_d$, it holds that $1 < c_d < 2$.
    The final complexity of our algorithm is $\bigOs{c_d^n}$, with $\sqrt[\leftroot{2}d]{2^d - 1} < c_d < 2$, since $p_d\left(\sqrt[\leftroot{2}d]{2^d - 1}\right) = -(2^d - 1)^{\frac{n-d-1}{d}} < 0$.
    Table~\ref{tab:exact_values} presents the branching factors for some values of $d$ for our two branching rules.\end{proof}

\begin{table}[!htb]
    \centering
    \begin{tabular}{c|ccccccc}
         $d$ & 1 & 2 & 3 & 4 & 5 & 6 & 7\\
         \hline
         B1 & 1.6180 & 1.8793 & 1.9583 & 1.9843 & 1.9937 & 1.9973 & 1.9988 \\
         B2 & 1.0000 & 1.7320 & 1.9129 & 1.9679 & 1.9873 & 1.9947 & 1.9977 \\

    \end{tabular}\medskip
    \caption{\centering Branching factors for some values of $d$.}
    \label{tab:exact_values}
\end{table}

\section{Parameterized algorithms and kernelization}
\label{sec:param}

In this section we focus on the parameterized complexity of \textsc{$d$-Cut}. More precisely, in Section~\ref{sec:crossing-edges} we consider as the parameter the number of edges crossing the cut, in Section~\ref{thm:algo-tw} the treewidth of the input graph, in Section~\ref{sec:kernelization} the distance to cluster (in particular, we provide a quadratic kernel), and in Section~\ref{sec:cocluster} the distance to co-cluster.

\subsection{Crossing edges}
\label{sec:crossing-edges}

In this section we consider as the parameter the maximum number of edges crossing the cut. In a nutshell, our approach is to use as a black box one of the algorithms presented by Marx et al.~\cite{marx_treewidth_reduction} for a class of separation problems. Their fundamental problem is \pname{$\mathcal{G}$-MinCut}, for a fixed class of graphs $\mathcal{G}$, which we state formally, along with their main result, below.

\paraprobl{$\mathcal{G}$-MinCut}{A graph $G$, vertices $s,t $, and an integer $k$.}{Is there an induced subgraph $H$ of $G$ with at most $k$ vertices such that $H \in \mathcal{G}$ and $H$ is an $s-t$ separator?}{The integer $k$.}{13.5}

\begin{theorem}[Theorem 3.1 in~\cite{marx_treewidth_reduction}]
    \label{thm:marx}
    If $\mathcal{G}$ is a decidable and hereditary graph class, \pname{$\mathcal{G}$-MinCut} is $\FPT$.
\end{theorem}

To be able to apply Theorem~\ref{thm:marx}, we first need to specify a graph class to which, on the line graph, our separators correspond. We must also be careful to guarantee that the removal of a separator in the line graph leaves non-empty components in the input graph. To accomplish that, for each $v \in V(G)$, we add a private clique of size $2d$ adjacent only to it, choose one arbitrary vertex $v'$ in each of them, and our algorithm will ask for the existence of a ``special'' separator of the appropriate size between every pair of chosen vertices of two distinct private cliques. We assume henceforth that these private cliques have been added to the input graph $G$.

For each integer $d \geq 1$, we define the graph class $\mathcal{G}_d$ as follows.

\begin{definition}
    A graph $H$ belongs to $\mathcal{G}_d$ if and only if its maximum clique size is at most $d$.
\end{definition}

Note that $\mathcal{G}_d$ is clearly decidable and hereditary for every integer $d \geq 1$.

\begin{lemma}
    \label{lem:cut_mincut}
    $G$ has a $d$-cut if and only if $L(G)$ has a vertex separator belonging to $\mathcal{G}_d$.
\end{lemma}

\begin{proof}
    Let $H = L(G)$, $(A, B)$ be a $d$-cut of $G$, and $F \subseteq V(H)$ be the set of vertices such that $e_{uv} \in F$ if and only if $u \in A$ and $v \in B$, or vice-versa.
    The fact that $F$ is a separator of $H$ follows directly from the hypothesis that $(A, B)$ is a cut of $G$.
    Now, to show that $H[F] \in \mathcal{G}_d$, suppose for contradiction that $H[F]$ contains a  clique $Q$ with more than $d$ vertices.
    That is, there are at least $d+1$ edges of $G$ that are pairwise intersecting and with one endpoint in $A$ and the other in $B$.
    Note, however, that for at least one of the parts, say $A$, there is also {\sl at most} one vertex with an edge in $Q \subseteq E(G)$, as otherwise there would be two non-adjacent vertices in the clique $Q \subseteq V(H)$.
    As such, $A$ has only one vertex and we conclude that every edge in $Q$ has an endpoint in $A$, but this, on the other hand, implies that $A$ has $d+1$ neighbors in $B$, contradicting the hypothesis that $(A, B)$ is a $d$-cut of $G$.

    For the converse, take a vertex separator $S \subseteq V(H)$ such that $H[S] \in \mathcal{G}_d$ and let $E_S$ be the edges of $G$ corresponding to $S$.
    Let $G'$ be the graph where each vertex corresponds to a connected component of $G - E_S$ and two vertices are adjacent if and only if there is an edge in $E_S$ between vertices of the respective components.
    Let $Q_r$ be an arbitrarily chosen connected component of $G - E_S$.
    Now, for each component at an odd distance from $Q_r$ in $G'$, add that component to $B$; all other components are placed in $A$.
    We claim that $(A, B)$ is a $d$-cut of $G$. Let $F \subseteq E_S$ be the set of edges with one endpoint in $A$ and the other in $B$.
    Note that $G - F$ is disconnected due to the construction of $A$ and $B$.
    If there is some $v \in A$ with more than $d$ neighbors in $B$, we obtain that there is some clique of equal size in $H[S]$, contradicting the hypothesis that this subgraph belongs to $\mathcal{G}_d$.
\end{proof}

\begin{theorem}\label{thm:FPT-crossing}
    For every $d \geq 1$, there is an $\FPT$ algorithm for \pname{$d$-Cut} parameterized by $k$, the maximum number of edges crossing the cut.
\end{theorem}

\begin{proof}
    For each pair of vertices $s,t \in V(G)$ that do not belong to the private cliques, our goal is to find a subset of vertices $S \subseteq V(L(G))$ of size at most $k$ that separates $s$ and $t$ such that $L(G)[S] \in \mathcal{G}_d$.
    This is precisely what is provided by Theorem~\ref{thm:marx}, and the correctness of this approach is guaranteed by Lemma~\ref{lem:cut_mincut}.
    Since we perform a quadratic number of calls to the  algorithm given by Theorem~\ref{thm:marx}, our algorithm still runs in $\FPT$ time.
\end{proof}

As to the running time of the $\FPT$ algorithm given by Theorem~\ref{thm:FPT-crossing}, the treewidth reduction technique of~\cite{marx_treewidth_reduction} relies on the construction of a monadic second order logic (MSOL) expression and Courcelle's Theorem~\cite{courcelle_theorem} to guarantee fixed-parameter tractability, and therefore it is hard to provide an explicit running time in terms of $k$.

\subsection{Treewidth}
\label{thm:algo-tw}

We proceed to present an algorithm for \pname{$d$-Cut} parameterized by the treewidth of the input graph that, in particular, improves the running time of the best known algorithm for \pname{Matching Cut}~\cite{matching_cut_structural}.
For the definitions of treewidth we refer to~\cite{treewidth,CyganFKLMPPS15}.
We state here an adapted definition of nice tree decomposition which shall be useful in our algorithm.


\begin{definition}{(Nice tree decomposition)}
    A tree decomposition $(T, \mathcal{B})$ of a graph $G$ is said to be \emph{nice} if it T is a tree rooted at an empty bag $r(T)$ and each of its bags is from one of the following four types:
    \begin{enumerate}
        \item \textit{Leaf node}: a leaf $x$ of $T$ with $|B_x| = 2$ and no children.
        \item \textit{Introduce node}: an inner node $x$ of $T$ with one child $y$ such that $B_x \setminus B_y = \{u\}$, for some $u \in V(G)$.
        \item \textit{Forget node}: an inner node $x$ of $T$ with one child $y$ such that $B_y \setminus B_x = \{u\}$, for some $u \in V(G)$.
        \item \textit{Join node}: an inner node $x$ of $T$ with two children $y,z$ such that $B_x = B_y = B_z$.
    \end{enumerate}
\end{definition}

In the next theorem, note that the assumption that the given tree decomposition is {\sl nice} is not restrictive, as any tree decomposition can be transformed into a nice one of the same width in polynomial time~\cite{Klo94}.

\begin{theorem}\label{thm:treewidth}
    For every integer $d \geq 1$, given a nice tree decomposition of $G$ of width $\tw(G)$, \pname{$d$-Cut} can be solved in time $\bigOs{2^{\tw(G)+1}(d+1)^{2\tw(G) + 2}}$.
\end{theorem}

\begin{proof}
    As expected, we will perform dynamic programming on a nice tree decomposition.
    For this proof, we denote a $d$-cut of $G$ by $(L, R)$ and suppose that we are given a total ordering of the vertices of $G$.
    Let $(T, \mathcal{B})$ be a nice tree decomposition of $G$ rooted at a node $r \in V(T)$.
    For a given node $x \in T$, an entry of our table is indexed by a triple $(A, \balpha, t)$, where $A \subseteq B_x$, $\balpha \in \left(\{0\} \cup [d]\right)^{\tw(G)+1}$, and $t$ is a binary value. Each coordinate $a_i$ of $\alpha$ indicates how many vertices {\sl outside} of $B_x$ the $i$-th vertex of $B_x$ has in the other side of the partition. More precisely, we denote by $f_x(A, \balpha, t)$ the binary value indicating whether or not $V(G_x)$ has a bipartition $(L_x,R_x)$ such that $L_x \cap B_x = A$, every vertex $v_i \in B_x$ has exactly $a_i$ neighbors in the other side of the partition $(L_x,R_x)$  outside of $B_x$, and both $L_x$ and $R_x$ are non-empty if and only if $t = 1$. Note that $G$ admits a $d$-cut if and only if $f_r(\emptyset, \boldsymbol{0}, 1)=1$.
    Figure~\ref{fig:treewidth} gives an example of an entry in the dynamic programming table and the corresponding solution on the subtree.



    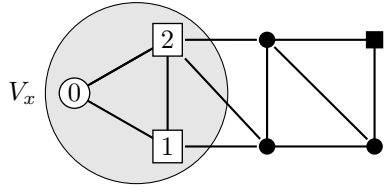
\begin{figure}[!htb]
        \centering
        \begin{tikzpicture}[rotate = 90]
                \GraphInit[unit=3,vstyle=Normal]
                \SetVertexNormal[Shape=circle, MinSize=3pt]
                \tikzset{VertexStyle/.append style = {inner sep = \inners, outer sep = \outers}}
                \SetVertexNoLabel
                \begin{scope}[rotate=90]
                    \draw[fill=gray!20] (0,0) circle (1.2);
                    \node at (1.5, 0) {$V_x$};
                    \grComplete[RA=0.816747, prefix=t]{3}
                    \SetVertexLabel
                    \Vertex[Node, L = {0}, Math]{t0}
                \end{scope}
                \begin{scope}[rotate=90]
                    \tikzset{VertexStyle/.append style = {inner sep = 3pt, shape = rectangle}}
                    \SetVertexLabel
                    \Vertex[Node, L = {2}, Math]{t2}
                    \Vertex[Node, L = {1}, Math]{t1}
                \end{scope}
                \begin{scope}[rotate=45, shift={(-1.71424, -1.71424)}]
                    \grCycle[RA=1, prefix=c]{4}
                    \Edge(t0)(t2)
                \end{scope}
                \begin{scope}[rotate=45, shift={(-1.71424, -1.71424)}]
                    \SetVertexNormal[Shape=circle, FillColor = black, MinSize=3pt]
                    \tikzset{VertexStyle/.append style = {shape = rectangle, inner sep = 3pt, outer sep = \outers}}
                    \Vertex[Node]{c3}
                \end{scope}
                \begin{scope}[rotate=45, shift={(-1.71424, -1.71424)}]
                    \SetVertexNormal[Shape=circle, FillColor = black, MinSize=3pt]
                    \tikzset{VertexStyle/.append style = {inner sep = 2pt, outer sep = \outers}}
                    \Vertex[Node]{c0}
                    \Vertex[Node]{c1}
                    \Vertex[Node]{c2}
                \end{scope}
                \Edge(c0)(t2)
                \Edge(c1)(t2)
                \Edge(c1)(t1)
                \Edge(c0)(c2)

        \end{tikzpicture}
        \caption{\centering Example for $d=3$ of dynamic programming state and corresponding solution on the subtree. Squared (circled) vertices belong to $A$ ($B$). Numbers indicate the respective value of $\alpha_i$.\label{fig:treewidth}}
    \end{figure}

    We say that an entry $(A, \balpha, t)$ for a node $x$ is \tdef{valid} if for every $v_i \in A$, $|N(v_i) \cap (B_x \setminus A)| + a_i \leq d$, for every $v_j \in B_x \setminus A$, $|N(v_i) \cap A| + a_j \leq d$, and if $B_x \setminus A \neq \emptyset$ then $t = 1$; otherwise the entry is \tdef{invalid}. Moreover, note that if $f_x(A, \balpha, t)=1$, the corresponding bipartition $(L_x,R_x)$ of $V(G_x)$ is a $d$-cut if and only if $(A, \balpha, t)$ is valid and $t = 1$.

    We now explain how the entries for a node $x$ can be computed, assuming recursively that the entries for their children have been already computed. We distinguish the four possible types of nodes. Whenever $(A, \balpha, t)$ is invalid or absurd (with, for example, $a_i < 0$) we define $f_x(A, \alpha, t)$ to be $0$, and for simplicity we will not specify this in the equations stated below.

    \begin{itemize}
        \item Leaf node: Since $|B_x| = 2$, for every $A \subseteq B_x$, we can set $f_x(A, \boldsymbol{0}, t) = 1$ with $t = 1$ if and only if $B_x \setminus A \neq \emptyset$.
        These are all the possible partitions of $B_x$, taking $\bigO{1}$ time to be computed.

        \item Introduce node: Let $y$ be the child of $x$ and $B_x \setminus B_y = \{v_i\}$.
        The transition is given by the following equation, where $\balpha^*$ has entries equal to $\balpha$ but without the coordinate corresponding to $v_i$.
        If $a_i > 0$, $f_x(A, \balpha, t)$ is invalid since $v_i$ has no neighbors in $G_x - B_x$.
        \[
     f_x(A, \balpha, t)=\left\{
                \begin{array}{ll}
                  f_y(A \setminus \{v\}, \balpha^*, t), & \text{if $A = B_x$ or $A = \emptyset$.}\\
                  \max_{t' \in \{0,1\}} f_y(A \setminus \{v\}, \balpha^*, t'), & \text{otherwise.}
                \end{array}
              \right.
       \]

        For the first case, $G_x$ has a bipartition (which will also be a $d$-cut if $t=1$)  represented by $(A, \balpha, t)$ only if $G_y$ has a bipartition ($d$-cut), precisely because, in both $G_x$ and $G_y$, the entire bag is in one side of the cut.
        For the latter case, if $G_y$ has a bipartition, regardless if it is a $d$-cut or not, $G_x$ has a $d$-cut 
        because $B_x$ is not contained in a single part of the cut, unless the entry is invalid.
        The computation for each of these nodes takes $\bigO{1}$ time per entry.

        \item Forget node: Let $y$ be the child of $x$ and $B_y \setminus B_x = \{v_i\}$.
        In the next equation, $\balpha'$ has the same entries as $\balpha$ with the addition of entry $a_i$ corresponding to $v_i$ and, for each $v_j \in A \cap N(v_i)$, $a_j' = a_j - 1$.
        Similarly, for $\balpha''$, for each $v_j \in (B_x \setminus A) \cap N(v_i)$, $a_j'' = a_j - 1$.
        \begin{equation*}
            f_x(A, \balpha, t) = \max_{a_i \in \{0\} \cup [d]}\ \max \{ f_y(A, \balpha', t),\  f_y(A \cup \{v_i\}, \balpha'', t)\}.
        \end{equation*}

        Note that $\balpha'$ and $\balpha''$ take into account the forgetting of $v_i$; its neighbors get an additional neighbor outside of $B_x$ that is in the other side of the bipartition.
        Moreover, since we inspect the entries of $y$ for every possible value of $a_i$, if at least one of them represented a feasible bipartition of $G_y$, the corresponding entry on $f_y(\cdot)$ would be non-zero and, consequently, $f_x(A, \balpha, t)$ would also be non-zero.
        Computing an entry for a forget node takes $\bigO{d}$ time.

        \item Join node: Finally, for a join node $x$ with children $y$ and $z$, a \tdef{splitting} of $\balpha$ is a pair $\balpha_y, \balpha_z$ such that for every coordinate $a_j$ of $\balpha$, it holds that the sum of $j$-th coordinates of $\balpha_y$ and $\balpha_z$ is equal to $a_j$. The set of all splittings is denoted by $S(\balpha)$ and has size $\bigO{(d+1)^{\tw(G)+1}}$.
        As such, we define our transition function as follows.
        \begin{equation*}
            f_x(A, \balpha, t) = \max_{t \leq t_y + t_z \leq 2t}\ \max_{S(\balpha)} f_y(A, \balpha_y, t_y) \cdot f_z(A, \balpha_z, t_z).
        \end{equation*}

        The condition $t \leq t_y + t_z \leq 2t$ enforces that, if $t = 1$, at least one of the graphs $G_y, G_z$ must have a $d$-cut; otherwise, if $t = 0$, neither of them can.
        When iterating over all splittings of $\balpha$, we are essentially testing all possible counts of neighbors outside of $B_y$ such that there exists some entry for node $z$ such that $\balpha_y + \balpha_z = \balpha$.
        Finally, $f_x(A, \balpha, t)$ is feasible if there is at least one splitting and $t_y, t_z$ such that both $G_y$ and $G_z$ admit a bipartition.
        This node type, which is the bottleneck of our dynamic programming approach, takes $\bigO{(d+1)^{\tw(G)+1}}$ time per entry.
    \end{itemize}

    Consequently, since we have $\bigO{\tw(G)} \cdot n$ nodes in a nice tree decomposition, spend $\bigO{\tw(G)^2}$ to detect an invalid entry, have $\bigO{2^{\tw(G) +1}(d+1)^{\tw(G)+1}}$ entries per node, each taking at most $\bigO{(d+1)^{\tw(G)+1}}$ time to be computed, our algorithm runs in time $\bigO{\tw(G)^32^{\tw(G)+1}(d+1)^{2\tw(G)+2}\cdot n}$, as claimed.
\end{proof}

From Theorem~\ref{thm:treewidth} we immediately get the following corollary, which improves over the algorithm given by Aravind et al.~\cite{matching_cut_structural}.

\begin{corollary}
   Given a nice tree decomposition of $G$ of width $\tw(G)$, \pname{Matching Cut} can be solved in time  $\bigOs{8^{\tw(G)}}$.
\end{corollary}

\subsection{Kernelization and distance to cluster}
\label{sec:kernelization}

The proof of the following theorem consists of a simple generalization to every $d \geq 1$ of the construction given by Komusiewicz et al.~\cite{matching_cut_ipec}  for $d=1$.

\begin{theorem}\label{thm:no-kernel}
    For any fixed $d \geq 1$, \pname{$d$-Cut} does not admit a polynomial kernel when simultaneously parameterized by $k$, $\Delta$, and $\tw(G)$, unless $\NP \subseteq \coNP/\poly$.
\end{theorem}

\begin{proof}
    We show that the problem cross-composes into itself.
    Start with $t$ instances $G_1, \dots, G_t$ of \pname{$d$-Cut}.
    First, pick an arbitrary vertex $v_i \in V(G_i)$, for each $i \in [t]$.
    Second, for $i \in [t-1]$,  add a copy of $K_{2d}$, call it $K(i)$, every edge between $v_i$ and $K(i)$, and every edge between $K(i)$ and $v_{i+1}$.
    This concludes the construction of $G$, which  for $d=1$ coincides with that presented by Komusiewicz et al.~\cite{matching_cut_ipec}.

    Suppose that $(A, B)$ is a $d$-cut of some $G_i$ and that $v_i \in A$.
    Note that $(G \setminus B, B)$ is a $d$-cut of $G$ since the only edges in the cut are those between $A$ and $B$.
    For the converse, take some $d$-cut $(A, B)$ of $G$ and note that every vertex in the set $\{v_t\} \bigcup_{i \in [t-1]}\{v_i\} \cup K(i)$ is contained in the same side of the partition, say $A$.
    Since $B \neq \emptyset$, for any edge $uv$ crossing the cut, there is some $i$ such that $\{u,v\} \in V(G_i)$, which implies that there is some $i$ (possibly more than one) such that $(A \cap V(G_i), B \cap V(G_i))$ must also be a $d$-cut of $G_i$.

    That the treewidth, maximum degree, and number of edges crossing the partition are bounded by $n$, the maximum number of vertices of the graphs $G_i$, is a trivial observation.
\end{proof}

We now proceed to show that \pname{$d$-Cut} admits a polynomial kernel when parameterizing by the \tdef{distance to cluster} parameter, denoted by $\dc$.
A \tdef{cluster graph} is a graph such that every connected component is a clique; the \emph{distance to cluster} of a graph $G$ is the minimum number of vertices we must remove from $G$ to obtain a cluster graph.
Our results are heavily inspired by the work of Komusiewicz et al.~\cite{matching_cut_ipec}.
Indeed, most of our reduction rules are natural generalizations of theirs. However, we need some extra observations and rules that only apply for $d \geq 2$, such as Rule~\ref{rule:pattern_removal}.


We denote by $U = \{U_1, \dots, U_t\}$ a set of vertices such that $G - U$ is a cluster graph, and each $U_i$ is called a \tdef{monochromatic part} or \tdef{monochromatic set} of $U$, and we will maintain the invariant that these sets are indeed monochromatic. Initially, we set each $U_i$ as a singleton.
In order to simplify the analysis of our instance, for each $U_i$ of size at least two, we will have a private clique of size $2d$ adjacent to every vertex of $U_i$, which we call $X_i$.
The \tdef{merge} operation between $U_i$ and $U_j$ is the following modification: delete $X_i \cup X_j$, set $U_i$ as $U_i \cup U_j$, $U_j$ as empty, and add a new clique of size $2d$, $X_{i,j}$, which is adjacent to every element of the new $U_i$.
We say that an operation is \textit{safe} if the resulting instance  is a $\YES$ instance if and only if the original instance was.

\begin{observation}
    If $U_i \cup U_j$ is monochromatic, merging $U_i$ and $U_j$ is safe.
\end{observation}


It is worth mentioning that the second case of the following rule is not needed in the corresponding rule in~\cite{matching_cut_ipec}; we need it here to prove the safeness of Rules~\ref{rule:super_small} and~\ref{rule:pattern_removal}.

\begin{rrule}
    \label{rule:trivial}
    Suppose that $G - U$ has some cluster $C$ such that
    \begin{enumerate}
        \item $(C, V(G) \setminus C)$ is a $d$-cut, or
        \item $|C| \leq 2d$ and there is $C' \subseteq C$ such that $(C', G \setminus C')$ is a $d$-cut.
    \end{enumerate}
    Then output $\YES$.
\end{rrule}

After applying Rule~\ref{rule:trivial}, for every cluster $C$, $C$  has some vertex with at least $d+1$ neighbors in $U$, or there is some vertex of $U$ with $d+1$ neighbors in $C$.
Moreover, note that no cluster $C$ with at least $2d+1$ vertices can be partitioned in such a way that one side of the cut is composed only by a proper subset of vertices of $C$.

The following definition is a natural generalization of the definition of the set $N^2$ given by Komusiewicz et al.~\cite{matching_cut_ipec}.
Essentially, it enumerates some of the cases where a vertex, or set of vertices, is monochromatic, based on its relationship with $U$.
However, there is a crucial difference that keeps us from achieving equivalent bounds both in terms of running time and size of the kernel, and which makes the analysis and some of the rules more complicated than in~\cite{matching_cut_ipec}.
Namely, for a vertex to be forced into a particular side of the cut, it must have at least $d+1$ neighbors in that side; moreover, a vertex of $U$ being adjacent to $2d$ vertices of a cluster $C$ implies that $C$ is monochromatic.
Only if $d=1$, i.e., when we are dealing with matching cuts, the equality $d+1 = 2d$ holds.
This gap between $d+1$ and $2d$ is the main difference between our kernelization algorithm for general $d$ and the one shown in~\cite{matching_cut_ipec} for \pname{Matching Cut}, and the main source of the differing complexities we obtain. In particular, for $d=1$ the fourth case of the following definition is a particular case of the third one, but this is not true anymore for $d \geq 2$.
Figure~\ref{fig:n2d} illustrates the set of vertices introduced in Definition~\ref{def:n2d}.

\begin{definition}
    \label{def:n2d}
    For a monochromatic part $U_i \subseteq U$, let $N^{2d}(U_i)$ be the set of vertices $v \in V(G) \setminus U$ for which at least one of the following holds:

    \begin{enumerate}
        \item $v$ has at least $d+1$ neighbors in $U_i$.
        \item $v$ is in a cluster $C$ of size at least $2d+1$ in $G - U$ such that there is some vertex of $C$ with at least $d+1$ neighbors in $U_i$.
        \item $v$ is in a cluster $C$ of $G - U$ and some vertex in $U_i$ has $2d$ neighbors in $C$.
        \item $v$ is in a cluster $C$ of $G - U$ of size at least $2d+1$ and some vertex in $U_i$ has $d+1$ neighbors in $C$.
    \end{enumerate}
\end{definition}

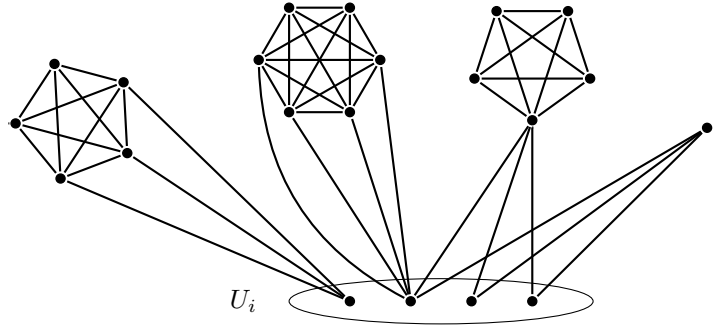
\begin{figure}[!htb]
        \centering
        \begin{tikzpicture}[scale=1, rotate = 180]
            \GraphInit[unit=3,vstyle=Normal]
            \SetVertexNormal[Shape=circle, FillColor=black, MinSize=1pt]
            \tikzset{VertexStyle/.append style = {inner sep = 1.2, outer sep = \outers}}
            \SetVertexNoLabel
            \node at (2.6, 0) {$U_i$};
            \draw (0,0) ellipse (2cm and 0.3cm);
            \Vertex[x=-1.2, y=0]{x0}
            \Vertex[x=-0.4, y=0]{x1}
            \Vertex[x=0.4, y=0]{x2}
            \Vertex[x=1.2, y=0]{x3}

            \Vertex[x=-3.5, y=-2.3]{c1}
            \Edge(c1)(x0)
            \Edge(c1)(x1)
            \Edge(c1)(x2)

            \begin{scope}[scale=0.8, shift={(-1.5,-4)}]
                \tikzset{VertexStyle/.append style = {inner sep = 1.2pt, outer sep = \outers}}
                \begin{scope}[rotate=90]
                    \grComplete[RA=1, prefix=p]{5}
                \end{scope}
                \Edge(p0)(x0)
                \Edge(p0)(x1)
                \Edge(p0)(x2)
            \end{scope}

            \begin{scope}[scale=0.8, shift={(2,-4)}]
                \tikzset{VertexStyle/.append style = {inner sep = 1.2pt, outer sep = \outers}}
                \begin{scope}[rotate=60]
                    \grComplete[RA=1, prefix=t]{6}
                \end{scope}
                \Edge(x2)(t0)
                \Edge(x2)(t1)
                \Edge[style = bend left](x2)(t5)
                \Edge(x2)(t2)
            \end{scope}

            \begin{scope}[scale=0.8, shift={(6,-3)}]
                \tikzset{VertexStyle/.append style = {inner sep = 1.2pt, outer sep = \outers}}
                \begin{scope}[rotate=75]
                    \grComplete[RA=1, prefix=q]{5}
                \end{scope}
                \Edge(x3)(q0)
                \Edge(x3)(q1)
                \Edge(x3)(q2)
            \end{scope}
        \end{tikzpicture}
        \caption{The four cases that define membership in $N^{2d}(U_i)$ for $d = 2$. \label{fig:n2d}}
    \end{figure}

\begin{observation}
    For every monochromatic part $U_i$, $U_i \cup N^{2d}(U_i)$ is monochromatic.
\end{observation}

The next rules aim to increase the size of monochromatic sets.
In particular, Rule~\ref{rule:transitivity} translates the transitivity of the monochromatic property, while Rule~\ref{rule:attraction} identifies a case where merging the monochromatic sets is inevitable.

\begin{rrule}
    \label{rule:transitivity}
    If $N^{2d}(U_i) \cap N^{2d}(U_j) \neq \emptyset$, merge $U_i$ and $U_j$.
\end{rrule}

\begin{rrule}
    \label{rule:attraction}
    If there there is a set of $2d+1$ vertices $L \subseteq V(G)$ with two common neighbors $u,u'$ such that $u \in U_i$ and $u' \in U_j$, merge $U_i$ and $U_j$.
\end{rrule}

\begin{sproof}{\ref{rule:attraction}}
    Suppose that in some $d$-cut $(A, B)$, $u \in A$ and $u' \in B$, this implies that at most $d$ elements of $L$ are in $A$ and at most $d$ are in $B$, which is impossible since $|L| = 2d+1$.
\end{sproof}

We say that a cluster is \textit{small} if it has at most $2d$ vertices, and \textit{big} otherwise.
Moreover, a vertex in a cluster is \textit{ambiguous} if it has neighbors in more than one $U_i$.
A cluster is \tdef{ambiguous} if it has an ambiguous vertex, and \tdef{fixed} if it is contained in some $N^{2d}(U_i)$.

\begin{observation}
    \label{obs:fix_amb}
    If $G$ is reduced by Rule~\ref{rule:trivial}, every big cluster is ambiguous or fixed.
\end{observation}

\begin{proof}
    Since Rule~\ref{rule:trivial} cannot be applied, every cluster $C$ has either one vertex $v$ with at least $d + 1$ neighbors in $U$ or there is some vertex of a set $U_i$ with $d + 1$ neighbors in $C$.
    In the latter case, by applying the fourth case in the definition of $N^{2d}(U_i)$, we conclude that $C$ is fixed.
    In the former case, either $v$ has $d+1$ neighbors in the same $U_i$, in which case $C$ is fixed, or its neighborhood is spread across multiple monochromatic sets, and so $v$ and, consequently, $C$ are ambiguous.
\end{proof}

Our next goal is to bound the number of vertices outside of $U$.

\begin{rrule}
    \label{rule:fusion}
    If there are two clusters $C_1, C_2$ contained in some $N^{2d}(U_i)$, then add every edge between $C_1$ and $C_2$.
\end{rrule}

\begin{sproof}{\ref{rule:fusion}}
    It follows directly from the fact that $C_1 \cup C_2$ is a larger cluster, $C_1 \cup C_2 \subseteq N^{2d}(U_i)$, and that adding edges between vertices of a monochromatic set preserves the existence of a $d$-cut.
\end{sproof}

The next lemma follows from the pigeonhole principle and exhaustive application of Rule~\ref{rule:fusion}.

\begin{lemma}
    \label{lem:fixed_clusters}
    If $G$ has been reduced by Rules~\ref{rule:trivial} through~\ref{rule:fusion}, then $G$ has $\bigO{|U|}$ fixed clusters.
\end{lemma}

\begin{rrule}
    \label{rule:shrink}
    If there is some cluster $C$ with at least $2d+2$ vertices such that there is some $v \in C$ with no neighbors in $U$, remove $v$ from $G$.
\end{rrule}

\begin{sproof}{\ref{rule:shrink}}
    That $G$ has a $d$-cut if and only if $G - v$ has a $d$-cut follows directly from the hypothesis that $C$ is monochromatic in $G$ and the fact that $|C \setminus \{v\}| \geq 2d + 1$ implies that $C \setminus \{v\}$ is monochromatic in $G - v$.
\end{sproof}

By Rule~\ref{rule:shrink}, we now have the additional property that, if $C$ has more than $2d+1$ vertices, all of them have at least one neighbor in $U$. The next rule provides a uniform structure between a big cluster $C$ and the sets $U_i$ such that $C \subseteq N^{2d}(U_i)$.

\begin{rrule}
    \label{rule:normalization1}
    If a cluster $C$ has at least $2d+1$ elements and there is some $U_i$ such that $C \subseteq N^{2d}(U_i)$, remove all edges between $C$ and $U_i$, choose $u \in U_i$, $\{v_1, \dots, v_{d+1}\} \subseteq C$ and add the edges $\{uv_i\}_{i \in [d+1]}$ to $G$.
\end{rrule}

\begin{sproof}{\ref{rule:normalization1}}
    Let $G'$ be the graph obtained after the operation is applied.
    If $G$ has some $d$-cut $(A,B)$, since $U_i \cup N^{2d}(U_i)$ is monochromatic, no edge between $U_i$ and $C$ crosses the cut, so $(A,B)$ is also a $d$-cut of $G'$.
    For the converse, take a $d$-cut $(A', B')$ of $G'$.
    Since $C$ has at least $2d+1$ vertices and there is some $u \in U_i$ such that $|N(u) \cap C| = d+1$, $C \in N^{2d}(U_i)$ in $G'$.
    Therefore, no edge between $C$ and $U_i$ crosses the cut and $(A', B')$ is also a $d$-cut of $G$.
\end{sproof}

We have now effectively bounded the number of vertices in big clusters by a polynomial in $U$, as shown below.

\begin{lemma}
    \label{lem:bound1}
    If $G$ has been reduced by Rules~\ref{rule:trivial} through~\ref{rule:normalization1}, then $G$ has $\bigO{d|U|^2}$ ambiguous vertices and $\bigO{d|U|^2}$ big clusters, each with $\bigO{d|U|}$ vertices.
\end{lemma}

\begin{proof}
    To show the bound on the number of ambiguous vertices, take any two vertices $u \in U_i$, $u' \in U_j$.
    Since we have $\binom{|U|}{2}$ such pairs, if we had at least $(2d + 1)\binom{|U|}{2}$ ambiguous vertices, by the pigeonhole principle, there would certainly be $2d+1$ vertices in $V \setminus U$ that are adjacent to one pair, say $u$ and $u'$.
    This, however, contradicts the hypothesis that Rule~\ref{rule:attraction} has been applied, and so we have $\bigO{d|U|^2}$ ambiguous vertices.

    The above discussion, along with Lemma~\ref{lem:fixed_clusters} and Observation~\ref{obs:fix_amb}, imply that the number of big clusters is $\bigO{d|U|^2}$.
    For the bound on their sizes, take some cluster $C$ with at least $2d + 2$ vertices.
    Due to the application of Rule~\ref{rule:shrink}, every vertex of $C$ has at least one neighbor in $U$.
    Moreover, there is at most one $U_i$ such that $C \subseteq N^{2d}(U_i)$, otherwise we would be able to apply Rule~\ref{rule:transitivity}.

    Suppose first that there is such a set $U_i$.
    By Rule~\ref{rule:normalization1}, there is only one $u \in U_i$ that has neighbors in $C$; in particular, it has $d+1$ neighbors.
    Now, every $v \in U_j$, for every $j\neq i$, has at most $d$ neighbors in $C$, otherwise $C \subseteq N^{2d}(U_j)$ and Rule~\ref{rule:transitivity} would have been applied.
    Therefore, we conclude that $C$ has at most $(d+1) +  \sum_{v \in U \setminus U_i} |N(u) \cap C| \leq (d+1) + d|U| \in \bigO{d|U|}$ vertices.

    Finally, suppose that there is no $U_i$ such that $C \subseteq N^{2d}(U_i)$.
    A similar analysis from the previous case can be performed: every $u \in U_i$ has at most $d$ neighbors in $C$, otherwise $C \subseteq N^{2d}(U_i)$ and we conclude that $C$ has at most $\sum_{v \in U} |N(u) \cap C| \leq d|U| \in \bigO{d|U|}$ vertices.
\end{proof}

We are now left only with an unbounded number of small clusters.
A cluster $C$ is \tdef{simple} if it is not ambiguous, that is, if for each $v \in C$, $v$ has neighbors in a single $U_i$.
Otherwise, $C$ is ambiguous and, because of Lemma~\ref{lem:bound1}, there are at most $\bigO{d|U|^2}$ such clusters.
As such, for a simple cluster $C$ and a vertex $v \in C$, we denote by $U(v)$ the monochromatic set of $U$ to which $v$ is adjacent.

\begin{rrule}
    \label{rule:super_small}
    If $C$ is a simple cluster with at most $d+1$ vertices, remove $C$ from $G$.
\end{rrule}

\begin{sproof}{\ref{rule:super_small}}
    Let $G' = G - C$.
    Suppose $G$ has a $d$-cut $(A,B)$ and note that $A \nsubseteq C$ and $B \nsubseteq C$ since Rule~\ref{rule:trivial} does not apply.
    This implies that $(A \setminus C, B \setminus C)$ is a valid $d$-cut of $G'$.
    For the converse, take a $d$-cut $(A', B')$ of $G'$, define $C_A = \{v \in C \mid U(v) \subseteq A\}$, and define $C_B$ similarly; we claim that $(A' \cup C_A, B' \cup C_B)$ is a $d$-cut of $G$.
    To see that this is the case, note that each vertex of $C_A$ (resp. $C_B$) has at most $d$ edges to $C_B$ (resp. $C_A$) and, since $C$ is simple, $C_A$ (resp. $C_B$) has no other edges to $B'$ (resp. $A'$).
\end{sproof}

After applying the previous rule, every cluster $C$ not yet analyzed has size $d+2 \leq |C| \leq 2d$ which, in the case of the \pname{Matching Cut} problem, where $d=1$, is empty.
To deal with these clusters, given a $d$-cut $(A, B)$, we say that a vertex $v$ is in its \tdef{natural assignment} if $v \cup U(v)$ is in the same side of the cut; otherwise the vertex is in its \tdef{unnatural assignment}.
Similarly, a cluster is \tdef{unnaturally assigned} if it has an unnaturally assigned vertex, otherwise it is \tdef{naturally assigned}.

\begin{observation}
    \label{obs:constrained_unnatural}
    Let $\mathcal{C}$ be the set of all simple clusters with at least $d+2$ and no more than $2d$ vertices, and $(A,B)$ a partition of $V(G)$.
    If there are $d|U|+1$ edges $uv$, $v \in C \in \mathcal{C}$ and $u \in U$, such that $uv$ is crossing the partition, then $(A,B)$ is not a $d$-cut.
\end{observation}

\begin{proof}
    Since there are $d|U| + 1$ edges crossing the partition between $\mathcal{C}$ and $U$, there must be at least one $u \in U$ with $d+1$ neighbors in the other set of the partition.
\end{proof}

\begin{corollary}
    \label{cor:constrained_unnatural}
    In any $d$-cut of $G$, there are at most $d|U|$ unnaturally assigned vertices.
\end{corollary}


Our next lemma limits how many clusters in $\mathcal{C}$ relate in a similar way to $U$; we say that two simple clusters $C_1, C_2$ have the same \tdef{pattern} if they have the same size $s$ and there is a total ordering of $C_1$ and another of $C_2$ such that, for every $i \in [s]$, $v_i^1 \in C_1$ and $v_i^2 \in C_2$ satisfy $U(v_i^1) = U(v_i^2)$.
Essentially, clusters that have the same pattern have neighbors in exactly the same monochromatic sets of $U$ and the same multiplicity in terms of how many of their vertices are adjacent to a same monochromatic set $U_i$. Note that the actual neighborhoods in the sets $U_i$'s do not matter in order for two clusters to have the same pattern.
Figure~\ref{fig:assignment} gives an example of a maximal set of unnaturally assigned clusters; that is, any other cluster with the same pattern as the one presented must be naturally assigned, otherwise some vertex of $U$ will violate the $d$-cut property.

\begin{figure}[!htb]
        \centering
        \begin{tikzpicture}[rotate = 90]
                \GraphInit[unit=3,vstyle=Normal]
                \SetVertexNormal[Shape=circle, FillColor = black, MinSize=3pt]
                \tikzset{VertexStyle/.append style = {inner sep = \inners, outer sep = \outers}}
                \SetVertexNoLabel
                \begin{scope}[shift={(-3, 0)}]
                    \tikzset{VertexStyle/.append style = {shape = rectangle, inner sep = 2pt}}
                    \node at (0, 2.6) {$U_2$};
                    \draw (-0.3, 2.3) rectangle (0.3, -2.3);
                    \Vertex[x = 0, y = 2]{u11}
                    \Vertex[x = 0, y = -2]{u12}
                \end{scope}
                \begin{scope}[shift={(3, 0)}]
                    \node at (0, 2.6) {$U_1$};
                    \draw (-0.3, 2.3) rectangle (0.3, -2.3);
                    \Vertex[x = 0, y = 2]{u21}
                    \Vertex[x = 0, y = -2]{u22}
                \end{scope}

                \begin{scope}[shift={(0, 3)}]
                    \grComplete[RA=0.6, prefix=q1]{6}
                    \begin{scope}
                        \tikzset{VertexStyle/.append style = {shape = rectangle, inner sep = 2pt}}
                        \Vertex[Node]{q10}
                        \Vertex[Node]{q11}
                        \Vertex[Node]{q15}
                    \end{scope}
                    \Edge(u11)(q13)
                    \Edge(u12)(q14)
                    \Edge(u21)(q10)
                    \Edge(u22)(q15)
                \end{scope}

                \begin{scope}[shift={(0, 1)}]
                    \grComplete[RA=0.6, prefix=q2]{6}
                    \begin{scope}
                        \tikzset{VertexStyle/.append style = {shape = rectangle, inner sep = 2pt}}
                        \Vertex[Node]{q20}
                        \Vertex[Node]{q21}
                        \Vertex[Node]{q25}
                    \end{scope}
                    \Edge(u11)(q22)
                    \Edge(u12)(q23)
                    \Edge(u22)(q20)
                    \Edge(u21)(q21)
                \end{scope}

                \begin{scope}[shift={(0, -1)}]
                    \grComplete[RA=0.6, prefix=q3]{6}
                    \begin{scope}
                        \tikzset{VertexStyle/.append style = {shape = rectangle, inner sep = 2pt}}
                        \Vertex[Node]{q30}
                        \Vertex[Node]{q31}
                        \Vertex[Node]{q35}
                    \end{scope}
                    \Edge(u12)(q34)
                    \Edge(u11)(q33)
                    \Edge(u21)(q30)
                    \Edge(u22)(q35)
                \end{scope}

                \begin{scope}[shift={(0, -3)}]
                    \grComplete[RA=0.6, prefix=q4]{6}
                    \begin{scope}
                        \tikzset{VertexStyle/.append style = {shape = rectangle, inner sep = 2pt}}
                        \Vertex[Node]{q40}
                        \Vertex[Node]{q41}
                        \Vertex[Node]{q45}
                    \end{scope}
                    \Edge(u11)(q42)
                    \Edge(u12)(q43)
                    \Edge(u22)(q40)
                    \Edge(u21)(q41)
                \end{scope}

        \end{tikzpicture}
        \caption{\centering Example for $d=4$ of a maximal set of unassigned clusters. Squared (resp. circled) vertices would be assigned to $A$ (resp. $B$).\label{fig:assignment}}
    \end{figure}
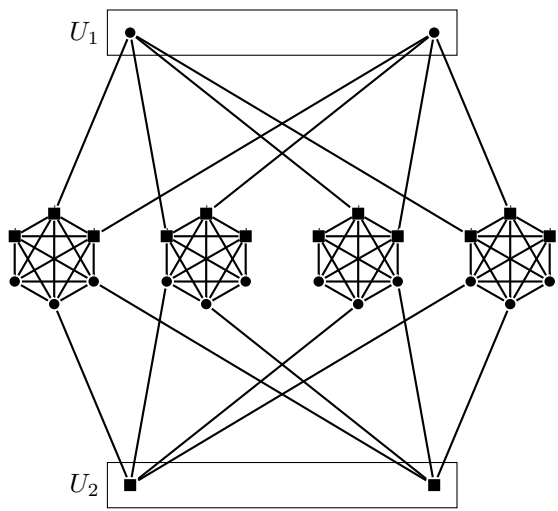


\begin{lemma}
    \label{lem:patterns}
    Let $\mathcal{C}^* \subseteq \mathcal{C}$ be a subfamily of simple clusters, all with the same pattern, with $|\mathcal{C}^*| > d|U| + 1$. Let $C$ be some cluster of $\mathcal{C}^*$, and $G' = G - C$. Then
    $G$ has a $d$-cut if and only if $G'$ has a $d$-cut.
\end{lemma}

\begin{proof}
    Since by Rule~\ref{rule:trivial} no subset of a small cluster is alone in a side of a partition and, consequently, $U$ intersects both sides of the partition, if $G$ has a $d$-cut, so does $G'$.

    For the converse, let $(A', B')$ be a $d$-cut of $G'$.
    First, by Corollary~\ref{cor:constrained_unnatural}, we know that at least one of the clusters of $\mathcal{C}^* \setminus \{C\}$, say $C_{{\sf n}}$, is naturally assigned.
    Since all the clusters in $\mathcal{C^*}$ have the same pattern, this guarantees that {\sl any} of the vertices of a naturally assigned cluster cannot have more than $d$ neighbors in the other side of the partition.

    Let $(A,B)$ be the bipartition of $V(G)$ obtained from $(A',B')$ such that $u \in C$ is in $A$ (resp. $B$) if and only if $U(u) \subseteq A$ (resp. $U(u) \subseteq B$); that is, $C$ is naturally assigned.
    Define $C_A = C \cap A$ and $C_B = C \cap B$.
    Because $|C| = |C_{{\sf n}}|$ and both belong to $\mathcal{C}^*$, we know that for every $u \in C_A$, it holds that $|N(u) \cap C_B| \leq d$; moreover, note that $N(u) \cap (B \setminus C) = \emptyset$. A symmetric analysis applies to every $u \in C_B$.
    This implies that no vertex of $C$ has additional neighbors in the other side of the partition outside of its own cluster and, therefore, $(A, B)$ is a $d$-cut of $G$.
\end{proof}

The safeness of our last rule follows directly from Lemma~\ref{lem:patterns}.

\begin{rrule}
    \label{rule:pattern_removal}
    If there is some pattern such that the number of simple clusters with that pattern is at least $d|U|+2$, delete all but $d|U|+1$ of them.
\end{rrule}

\begin{lemma}
    \label{lem:bound2}
    After exhaustive application of Rules~\ref{rule:trivial} through~\ref{rule:pattern_removal}, $G$ has $\bigO{d|U|^{2d}}$ small clusters and $\bigO{d^2|U|^{2d+1}}$ vertices in these clusters.
\end{lemma}

\begin{proof}
    By Rule~\ref{rule:super_small}, no small cluster with less than $d+2$ vertices remains in $G$.
    Now, for the remaining sizes, for each $d+2 \leq s \leq 2d$, and each pattern of size $s$, by Rule~\ref{rule:pattern_removal} we know that the number of clusters with $s$ vertices that have the same pattern is at most $d|U| + 1$.
    Since we have at most $|U|$ possibilities for each of the $s$ vertices of a cluster, we end up with $\bigO{|U|^{s}}$ possible patterns for clusters of size $s$.
    Summing all of them up, we get that we have $\bigO{|U|^{2d}}$ patterns in total, and since each one has at most $d|U| + 1$ clusters of size at most $2d$, we get that we have at most $\bigO{d^2|U|^{2d+1}}$ vertices in those clusters.
\end{proof}

The exhaustive application of all the above rules and their accompanying lemmas are enough to show that indeed, there is a polynomial kernel for \pname{$d$-Cut} when parameterized by distance to cluster.

\begin{theorem}
    When parameterized by distance to cluster $\dc(G)$, \pname{$d$-Cut} admits a polynomial kernel with $\bigO{d^2\dc(G)^{2d+1}}$ vertices that can be computed in $\bigO{d^4\dc(G)^{2d+1}(n+m)}$ time.
\end{theorem}

\begin{proof}
    The algorithm begins by finding a set $U$ such that $G - U$ is a cluster graph.
    Note that $|U| \leq 3\dc(G)$ since a graph is a cluster graph if and only if it has no induced path on three vertices: while there is some $P_3$ in $G$, we know that at least one its vertices must be removed, but since we don't know which one, we remove all three; thus, $U$ can be found in $\bigO{\dc(G)(n + m)}$ time.
    After the exhaustive application of Rules~\ref{rule:trivial} through~\ref{rule:pattern_removal}, by Lemma~\ref{lem:bound1}, $V(G) \setminus U$ has at most $\bigO{d^2\dc(G)^3}$ vertices in clusters of size at least $2d+1$.
    By Rule~\ref{rule:super_small}, $G$ has no simple cluster of size at most $d+1$.
    Ambiguous clusters of size at most $2d$, again by Lemma~\ref{lem:bound1}, also comprise only $\bigO{d^2\dc(G)^2}$ vertices of $G$.
    Finally, for simple clusters of size between $d+2$ and $2d$, Lemmas~\ref{lem:patterns} and~\ref{lem:bound2} guarantee that there are $\bigO{d^2\dc(G)^{2d+1}}$ vertices in small clusters and, consequently, this many vertices in $G$.

    As to the running time, first, computing and maintaining $N^{2d}(U_i)$ takes $\bigO{d\dc(G)n}$ time.
    Rule~\ref{rule:trivial} is applied only at the beginning of the kernelization, and runs in $\bigO{2^{2d}d(n + m)}$ time.
    Rules~\ref{rule:transitivity} and~\ref{rule:attraction} can both be verified in $\bigO{d\dc(G)^2(n + m)}$ time, since we are just updating $N^{2d}(U_i)$ and performing merge operations.
    Both are performed only $\bigO{\dc(G)^2}$ times, because we only have this many pairs of monochromatic parts.
    The straightforward application of Rule~\ref{rule:fusion} would yield a running time of $\bigO{n^2}$. However, we can ignore edges that are interior to clusters and only maintain which vertices belong together; this effectively allows us to perform this rule in $\bigO{n}$ time, which, along with its $\bigO{n}$ possible applications, yields a total running time of $\bigO{n^2}$ for this rule.
    Rule~\ref{rule:shrink} is directly applied in $\bigO{n}$ time; indeed, all of its applications can be performed in a single pass.
    Rule~\ref{rule:normalization1} is also easily applied in $\bigO{n + m}$ time. Moreover, it is only applied $\bigO{\dc(G)}$ times, since, by Lemma~\ref{lem:bound1}, the number of fixed clusters is linear in $\dc(G)$; furthermore, we may be able to reapply Rule~\ref{rule:normalization1} directly to the resulting cluster, at no additional complexity cost.
    The analysis for Rule~\ref{rule:super_small} follows the same argument as for Rule~\ref{rule:shrink}.
    Finally, Rule~\ref{rule:pattern_removal} is the bottleneck of our kernel, since it must check each of the possible $\bigO{\dc(G)^{2d}}$ patterns, spending $\bigO{n}$ time for each of them.
    Each pattern is only inspected once because the number of clusters in a pattern can no longer achieve the necessary bound for the rule to be applied once the excessive clusters are removed.
\end{proof}

In the next theorem we provide an $\FPT$ algorithm for \textsc{$d$-Cut} parameterized by the distance to cluster, running in time $\bigO{4^d(d+1)^{\dc(G)}2^{\dc(G)}\dc(G)n^2}$. Our algorithm is based on dynamic programming, and is considerably simpler than the one given by Komusiewicz  et al.~\cite{matching_cut_ipec} for $d=1$, which applies four reduction rules and an equivalent formulation as a 2-\textsc{SAT} formula. However, for $d=1$ our algorithm is slower, namely $\bigOs{4^{\dc( G )}}$ compared to  $\bigOs{2^{\dc( G )}}$.


Observe that minimum distance to cluster
sets and minimum distance to co-cluster sets can be computed in $1.92^{\dc( G )} \cdot \bigO{ n^2 }$
time and $1.92^{\dcc( G )} \cdot \bigO{n^2 }$ time, respectively~\cite{branching-cluster}. Thus, in the proofs of Theorems~\ref{thm:fpt_cluster} and~\ref{thm:fpt_cocluster} we can safely assume that we have these sets at hand.

\begin{theorem}
    \label{thm:fpt_cluster}
    For every integer $d \geq 1$, there is an algorithm that solves \pname{$d$-Cut} in time $\bigO{4^d(d+1)^{\dc(G)}2^{\dc(G)}\dc(G)n^2}$.
\end{theorem}

\begin{proof}
    Let $U$ be a set such that $G - U$ is a cluster graph, $\mathcal{Q} = \{Q_1, \dots, Q_p\}$ be the family of clusters of $G - U$ and $\mathcal{Q}_i = \bigcup_{i \leq j \leq p} Q_j$.
    Essentially, the following dynamic programming algorithm attempts to extend a given partition of $U$ in all possible ways by partitioning clusters, one at a time, while only keeping track of the degrees of vertices that belong to $U$.
    Recall that we do not need to keep track of the degrees of the cluster vertices precisely because $G - U$ has no edge between clusters.

    Formally, given a partition $U = A  \dcup B$, our table is a mapping $f : [p] \times \mathbb{Z}^{|A|} \times \mathbb{Z}^{|B|} \rightarrow \{0, 1\}$.
    Each entry is indexed by $(i, \bmd_A, \bmd_B)$, where $i \in [p]$, $\bmd_A$ is a $|A|$-dimensional vector with the $j$-th coordinate begin denoted by $\bmd_A[j]$; $\bmd_B$ is defined analogously.
    Our goal is to have $f(i, \bmd_A, \bmd_B) = 1$ if and only if there is a partition $(X, Y)$ of $U \cup \mathcal{Q}_i$ where $A \subseteq X$, $B \subseteq Y$ and $v_j \in A$ ($u_{\ell} \in B$) has at most $\bmd_A[j]$ ($\bmd_B[\ell]$) neighbors in $\mathcal{Q}_i$.

    We denote by $P_d(i, \bmd_A, \bmd_B)$ the set of all partitions $L  \dcup R = Q_i$ such that every vertex $v \in L$ has $d_{B \cup R}(v) \leq d$, every $u \in R$ has $d_{A \cup L}(u) \leq d$, every $v_j \in A$, $d_R(v_j) \leq \bmd_A[j]$ and every $u_{\ell} \in B$, $d_L(u_{\ell}) \leq \bmd_B[\ell]$; note that, due to this definition, $(L, R) \neq (R, L)$.
    In the following equations, which give the computations required to build our table, $\bmd_A(R)$ and $\bmd_B(L)$ are the updated values of the vertices of $A$ and $B$ after $R$ is added to $Y$ and $L$ to $X$, respectively.
    \begin{align}
        f(i, \bmd_A, \bmd_B) &= 0 \bigvee_{(L, R) \in P_d(i, \bmd_A, \bmd_B)} f(i+1,\bmd_A(R), \bmd_B(L))\label{eq:trans_dc}\\
        f(p, \bmd_A, \bmd_B) &= 1, \text{if and only if $P_d(i, \bmd_A, \bmd_B) \neq \emptyset$.}\label{eq:stop_dc}
    \end{align}

    We proceed to show the correctness of the above by induction.
    For the base case, i.e., when $|\mathcal{Q}| = p = 1$, we have that for $v_j \in A$ ($u_l \in B$),  $\bmd_A[j] = d - d_B(v_j)$ ($\bmd_B[j] = d - d_A(u_l)$) and a partition of $V(G)$ exists if and only if there is some partition $(L, R) \in P_d(1, \bmd_A, \bmd_B)$, where.
    This case is covered by equation~(\ref{eq:stop_dc}).

    So let $p > 1$ and $(i, \bmd_A, \bmd_B)$ be an entry of our table.
    First, if $|Q_i| \geq 2d+1$, $Q_i$ is monochromatic, which implies that $|P_d(i, \bmd_A, \bmd_B)| \leq 2$.
    Therefore, we may assume that, $|P_d(i, \bmd_A, \bmd_B)| \leq 2^{2d}$.
    $P_d(i, \bmd_A, \bmd_B) = \emptyset$ implies that any partition $(L, R)$ of $Q_i$ causes a vertex in $L$ ($R$) to have more than $d$ neighbors in $B \cup R$ ($A \cup L$), which is easily checked for in $\bigO{n|U|}$-time, or some vertex $v_j \in A$ ($u_l \in B$) has $d_{Y \cup R}(v_j) > \bmd_A[j]$ ($d_{X \cup L}(u_l) > \bmd_B[l]$).
    Either way, we have that no matter how we partition $\mathcal{Q}_i$, the available degree of some vertex is not enough, equation~(\ref{eq:trans_dc}) yields the correct answer.

    However, if $P_d(i, \bmd_A, \bmd_B) \neq \emptyset$, the subgraph induced by $U \cup \mathcal{Q}_i$ has a $d$-cut separating $A$ and $B$ and respecting the limits of $\bmd_A$ and $\bmd_B$ if and only if there is some $(L, R) \in P_d(i, \bmd_A, \bmd_B)$ such that $U \cup \mathcal{Q}_{i+1}$ has a $d$-cut and each vertex of $A$ ($B$) has the size of its neighborhood in $\mathcal{Q}_{i+1}$ bounded by the respective coordinate of $\bmd_A(R)$ ($\bmd_B(L)$).
    By the inductive hypothesis, there is such a partition of $\mathcal{Q}_{i+1}$ if and only if $f(i+1, \bmd_A(R), \bmd_B(L)) = 1$, concluding the proof of correctness.
    Clearly, there is a $d$-cut separating $A$ and $B$ if $f(1, \bmd_A, \bmd_B) = 1$ where for every $v_j \in A$ ($u_l \in B$),  $\bmd_A[j] = d - d_B(v_j)$ ($\bmd_B[j] = d - \dgr_A(u_{\ell})$).

    The complexity analysis is straightforward.
    Recalling that $|P_d(i, \bmd_A, \bmd_B)| \leq 2^{2d}$, we have that each $f(i, \bmd_A, \bmd_B)$ can be computed in time $\bigO{4^d|U|n}$ and, since we have $\bigO{(d+1)^{|A| + |B|}p} \in \bigO{(d+1)^{|U|}p}$, given a partition $(A, B)$ of $U$, we can decide if there is $d$-cut separating $A$ and $B$ in $\bigO{4^d(d+1)^{|U|}|U|n^2}$-time.
    To solve \pname{$d$-Cut} itself, we guess all $2^{|U|}$ partitions of $U$ and, since $|U| \in \bigO{\dc(G)}$, we obtain a total running time of $\bigO{4^d(d+1)^{\dc(G)}2^{\dc(G)}\dc(G)n^2}$.
\end{proof}

\subsection{Distance to co-cluster}
\label{sec:cocluster}

A graph is a \emph{co-cluster} graph if only if its the complement of a cluster graph; that is, if it is a complete multipartite graph.
Our next theorem complements the results of our previous section and shall help establish the membership  in $\FPT$ of \pname{$d$-Cut} parameterized by the vertex cover number.

\begin{theorem}
    \label{thm:fpt_cocluster}
    For every integer $d \geq 1$, there is an algorithm solving \pname{$d$-Cut} in time $\bigO{32^d2^{\dcc(G)}(d+1)^{\dcc(G)+d}(\dcc(G)+d)n^2}$.
\end{theorem}

\begin{proof}
    Let $U \subseteq V(G)$ be a set of $\bigO{\dcc(G)}$ vertices such that $G - U$ is a co-cluster graph with color classes $\varphi = \{F_1, \dots, F_t\}$.
    Define $\mathcal{F} = \bigcup_{i \in [t]} F_i$ and suppose we are given a $d$-cut $(A,B)$ of $G[U]$.
    First, note that if $t \geq 2d+1$, we have that some of the vertices of $\mathcal{F}$ form a clique of size $2d+1$, which is a monochromatic set; furthermore, every vertex $v \in \mathcal{F}$ but not in $Q$ has at least $d+1$ neighbors in $Q$.
    This implies that $Q \cup \{v\}$ is monochromatic and, thus, $\mathcal{F}$ is a monochromatic set.
    Checking if either $(A \cup \mathcal{F}, B)$ or $(A, B \cup \mathcal{F})$ is a $d$-cut can be done in $\bigO{n^2}$ time.

    If the above does not apply, we have that $t \leq 2d$.
    \begin{itemize}
        \item Case 1: If $|\mathcal{F}| \leq 4d$ we can just try to extend $(A,B)$ with each of the $2^{|\mathcal{F}|}$ bipartitions of $\mathcal{F}$ in $\bigO{16^dn^2}$ time.
    \end{itemize}

    So now, let $\varphi_1 \dcup \varphi_2 = \varphi$ be a bipartition of the color classes, $\mathcal{F}_i = \{v \in F_j \mid F_j \in \varphi_i\}$, and, for simplicity, suppose that $|\mathcal{F}_1|~\leq~|\mathcal{F}_2|$.

    \begin{itemize}
        \item Case 2: If $|\mathcal{F}_1| \geq d+1$ and $|\mathcal{F}_2| \geq 2d+1$, we know that there is a set $Q \subseteq \mathcal{F}$ forming a (not necessarily induced) complete bipartite subgraph $K_{d+1, 2d+1}$, which is a monochromatic set.
        Again, any $v \notin Q$ has at least $d_Q(v) \geq d+1$, from which we conclude that $Q \cup \{v\}$ is also monochromatic, implying that $\mathcal{F}$ is monochromatic.
    \end{itemize}

    If Case 2 is not applicable, either $|\mathcal{F}_1| \leq d$ and $|\mathcal{F}_2| \geq 2d+1$, or $|\mathcal{F}_2| \leq 2d$.
    For the latter, note that this implies $|\mathcal{F}| \leq 4d$, which would have been solved by Case 1.
    For the former, two cases remain:

    \begin{itemize}
        \item Case 3: Every $F_i \in \varphi_2$ has $|F_i| \leq 2d$. This implies that every $F \in \varphi$ has size bounded by $2d$ and that $|\mathcal{F}| \leq 4d^2$; we can simply try to extend $(A, B)$ with each of the $\bigO{2^{d^2}}$ partitions of $\mathcal{F}$, which can be done in $\bigO{2^{d^2}n^2}$ time.

        \item Case 4: There is some $F_i \in \varphi_2$ with $|F_i| \geq 2d + 1$.
        Its existence implies that $|\mathcal{F}| - |F_i| \leq d$, otherwise we would have concluded that $\mathcal{F}$ is a monochromatic set.
        Since $\mathcal{F} \setminus F_i$ has at most $d$ vertices, the set of its bipartitions has size bounded by $2^d$.
        So, given a bipartition $\mathcal{F}_A \dcup \mathcal{F}_B = \mathcal{F} \setminus F$, we define $A' := A \cup \mathcal{F}_A$ and $B' := B \cup \mathcal{F}_B$.
        Finally, note that $G \setminus (A' \cup B')$ is a cluster graph where every cluster is a single vertex; that is, $\dc(G) \leq \dcc(G) + d$.
        In this case, we can apply Theorem~\ref{thm:fpt_cluster}, and obtain the running time of $\bigO{4^d(d+1)^{\dcc(G)+d}(\dcc(G)+d)n^2}$; we omit the term $2^{\dcc(G) + d}$ since we already have an initial partial $d$-cut $(A', B')$.
    \end{itemize}

    For the total complexity of the algorithm, we begin by guessing the initial partition of $U$ into $(A,B)$, spending $\bigO{n^2}$ time for each of the $\bigO{2^{\dcc(G)}}$ possible bipartitions.
    If $t \geq 2d+1$ we give the answer in $\bigO{n^2}$ time.
    Otherwise, $t \leq 2d$.
    If $|\mathcal{F}| \leq 4d$, then we spend $\bigO{16^dn^2}$ time to test all partitions of $\mathcal{F}$ and return the answer.
    Else, for each of the $\bigO{4^d}$ partitions of $\varphi$, if one of them has a part with $d+1$ vertices and the other part has $2d+1$ vertices, we respond in $\bigO{n^2}$ time.
    Finally, for the last two cases, we either need $\bigO{2^{d^2}n^2}$ time, or $\bigO{8^d(d+1)^{\dcc(G)+d}(\dcc(G)+d)n^2}$.
    This yields a final complexity of $\bigO{32^d2^{\dcc(G)}(d+1)^{\dcc(G)+d}(\dcc(G)+d)n^2}$.
\end{proof}

Using Theorems~\ref{thm:fpt_cluster},~\ref{thm:fpt_cocluster}, and the relation $\tau(G) \geq \max\{\dc(G), \dcc(G)\}$~\cite{matching_cut_ipec}, we obtain fixed-parameter tractability of \pname{$d$-Cut} for the vertex cover number $\tau(G)$.

\begin{corollary}
    For every $d \geq 1$, \pname{$d$-Cut} parameterized by the vertex cover number is in $\FPT$.
\end{corollary}

\section{Concluding remarks}
\label{sec:concl}

We presented a series of algorithms and complexity results; many questions, however, remain open.
For instance, all of our algorithms have an exponential dependency on $d$ on their running times.
While we believe that such a dependency is an intrinsic property of \pname{$d$-cut}, we have no proof for this claim. Similarly, the existence of a {\sl uniform} polynomial kernel parameterized by the distance to cluster, i.e., a kernel whose degree does not depend on $d$, remains an interesting open question.

Also in terms of running time, we expect the constants in the base of the exact exponential algorithm to be improvable. However, exploring small structures that yield non-marginal gains as branching rules, as done by Komusiewicz et al.~\cite{matching_cut_ipec} for $d=1$ does not seem a viable approach, as the number of such structures appears to rapidly grow along with $d$.


The distance to cluster kernel is hindered by the existence of clusters of size between $d+2$ and $2d$, an obstacle that is not present in the \pname{Matching Cut} problem.
Aside from the extremal argument presented, we know of no way of dealing with them.
We conjecture that it should be possible to reduce the total kernel size from $\bigO{d^2\dc(G)^{2d+1}}$ to $\bigO{d^2\dc(G)^{2d}}$, matching the size of the smallest known kernel for \pname{Matching Cut}~\cite{matching_cut_ipec}.

We also leave open to close the gap between the known polynomial and $\NP$-hard cases in terms of maximum degree.
We showed that, if $\Delta(G) \leq d+2$ the problem is easily solvable in polynomial time, while for graphs with $\Delta(G) \geq 2d+2$, it is $\NP$-hard.
But what about the gap $d+3 \leq \Delta(G) \leq 2d+1$?
After much effort, we were unable to settle any of these cases.
In particular, we are very interested in \pname{2-Cut},  which has a single open case, namely when $\Delta(G) = 5$.
After some weeks of computation, we found no graph with more than 18 vertices and maximum degree five that had no $2$-cut, in agreement with the computational findings of Ban and Linial~\cite{internal_partition_regular6}.
Interestingly, all graphs on 18 vertices without a $2$-cut are either 5-regular or have a single pair of vertices of degree 4, which are actually adjacent.
In both cases, the graph is maximal in the sense that we cannot add edges to it while maintaining the degree constraints.
We recall the initial discussion about the \pname{Internal Partition} problem; closing the gap between the known cases for \textsc{$d$-Cut} would yield significant advancements on the former problem.

Finally, the smallest $d$ for which $G$ admits a $d$-cut may be an interesting additional parameter to be considered when more traditional parameters, such as treewidth, fail to provide $\FPT$ algorithms by themselves.
Unfortunately, by Theorem~\ref{thm:regular_nph}, computing this parameter is not even in $\XP$, but, as we have shown, it can be computed in $\FPT$ time under many different parameterizations.

\bibliographystyle{abbrv}
\bibliography{refs}
\begin{thebibliography}{10}

\bibitem{matching_cut_structural}
N.~R. Aravind, S.~Kalyanasundaram, and A.~S. Kare.
\newblock On structural parameterizations of the matching cut problem.
\newblock In {\em Proc. of the 11th International Conference on Combinatorial
  Optimization and Applications (COCOA)}, volume 10628 of {\em LNCS}, pages
  475--482, 2017.

\bibitem{internal_partition_regular6}
A.~Ban and N.~Linial.
\newblock Internal partitions of regular graphs.
\newblock {\em Journal of Graph Theory}, 83(1):5--18, 2016.

\bibitem{BodlaenderDFH09}
H.~L. Bodlaender, R.~G. Downey, M.~R. Fellows, and D.~Hermelin.
\newblock On problems without polynomial kernels.
\newblock {\em Journal of Computer and System Sciences}, 75(8):423--434, 2009.

\bibitem{cross_composition}
H.~L. Bodlaender, B.~M.~P. Jansen, and S.~Kratsch.
\newblock Cross-composition: {A} new technique for kernelization lower bounds.
\newblock In {\em Proc. of the 28th International Symposium on Theoretical
  Aspects of Computer Science (STACS)}, volume~9 of {\em LIPIcs}, pages
  165--176, 2011.

\bibitem{matching_cut_planar}
P.~S. Bonsma.
\newblock The complexity of the matching-cut problem for planar graphs and
  other graph classes.
\newblock {\em Journal of Graph Theory}, 62(2):109--126, 2009.

\bibitem{branching-cluster}
A.~Boral, M.~Cygan, T.~Kociumaka, and M.~Pilipczuk.
\newblock A fast branching algorithm for cluster vertex deletion.
\newblock {\em Theory of Computing Systems}, 58(2):357--376, 2016.

\bibitem{chvatal_matching_cut}
V.~Chv{\'a}tal.
\newblock Recognizing decomposable graphs.
\newblock {\em Journal of Graph Theory}, 8(1):51--53, 1984.

\bibitem{courcelle_theorem}
B.~Courcelle.
\newblock {The monadic second-order logic of graphs. I. Recognizable sets of
  finite graphs}.
\newblock {\em Information and computation}, 85(1):12--75, 1990.

\bibitem{CyganFKLMPPS15}
M.~Cygan, F.~V. Fomin, L.~Kowalik, D.~Lokshtanov, D.~Marx, M.~Pilipczuk,
  M.~Pilipczuk, and S.~Saurabh.
\newblock {\em Parameterized Algorithms}.
\newblock Springer, 2015.

\bibitem{DeVos09}
M.~DeVos.
\newblock \texttt{http://www.openproblemgarden.org/op/friendly\_partitions},
  2009.

\bibitem{Die10}
R.~Diestel.
\newblock {\em {Graph Theory}}, volume 173.
\newblock Springer-Verlag, 4th edition, 2010.

\bibitem{DF13}
R.~G. Downey and M.~R. Fellows.
\newblock {\em Fundamentals of Parameterized Complexity}.
\newblock Texts in Computer Science. Springer, 2013.

\bibitem{exact_exponential_algorithms}
F.~V. Fomin and D.~Kratsch.
\newblock {\em Exact Exponential Algorithms}.
\newblock Springer, 2010.

\bibitem{book-kernels}
F.~V. Fomin, D.~Lokshtanov, S.~Saurabh, and M.~Zehavi.
\newblock {\em {Kernelization: Theory of Parameterized Preprocessing}}.
\newblock Cambridge University Press, 2019.

\bibitem{FortnowS11}
L.~Fortnow and R.~Santhanam.
\newblock {Infeasibility of instance compression and succinct PCPs for {NP}}.
\newblock {\em Journal of Computer and System Sciences}, 77(1):91--106, 2011.

\bibitem{matching_cut_graham}
R.~L. Graham.
\newblock On primitive graphs and optimal vertex assignments.
\newblock {\em Annals of the New York academy of sciences}, 175(1):170--186,
  1970.

\bibitem{eth}
R.~Impagliazzo and R.~Paturi.
\newblock {On the complexity of $k$-SAT}.
\newblock {\em Journal of Computer and System Sciences}, 62(2):367--375, 2001.

\bibitem{internal_partition_triangle_free}
A.~Kaneko.
\newblock On decomposition of triangle-free graphs under degree constraints.
\newblock {\em Journal of Graph Theory}, 27(1):7--9, 1998.

\bibitem{Klo94}
T.~Kloks.
\newblock {\em Treewidth. Computations and Approximations}.
\newblock Springer-Verlag LNCS, 1994.

\bibitem{matching_cut_ipec}
C.~Komusiewicz, D.~Kratsch, and V.~B. Le.
\newblock {Matching Cut: Kernelization, Single-Exponential Time FPT, and Exact
  Exponential Algorithms}.
\newblock In {\em Proc. of the 13th International Symposium on Parameterized
  and Exact Computation (IPEC)}, volume 115 of {\em LIPIcs}, pages 19:1--19:13,
  2018.

\bibitem{matching_cut_tcs}
D.~Kratsch and V.~B. Le.
\newblock Algorithms solving the matching cut problem.
\newblock {\em Theoretical Computer Science}, 609:328--335, 2016.

\bibitem{matching_cut_diameter}
H.~Le and V.~B. Le.
\newblock On the complexity of matching cut in graphs of fixed diameter.
\newblock In {\em Proc. of the 27th International Symposium on Algorithms and
  Computation (ISAAC)}, volume~64 of {\em LIPIcs}, pages 50:1--50:12, 2016.

\bibitem{stable_cutset_line_graphs}
V.~B. Le and B.~Randerath.
\newblock On stable cutsets in line graphs.
\newblock {\em Theoretical Computer Science}, 301(1-3):463--475, 2003.

\bibitem{lovasz_hypergraph}
L.~Lov{\'a}sz.
\newblock Coverings and colorings of hypergraphs.
\newblock In {\em Proc. of the 4th Southeastern Conference of Combinatorics,
  Graph Theory, and Computing}, pages 3--12. Utilitas Mathematica Publishing,
  1973.

\bibitem{internal_partition_c4_free}
J.~Ma and T.~Yang.
\newblock {Decomposing $C_4$-free graphs under degree constraints}.
\newblock {\em Journal of Graph Theory}, 90(1):13--23, 2019.

\bibitem{marx_treewidth_reduction}
D.~Marx, B.~O'Sullivan, and I.~Razgon.
\newblock Treewidth reduction for constrained separation and bipartization
  problems.
\newblock In {\em Proc. of the 27th International Symposium on Theoretical
  Aspects of Computer Science, (STACS)}, volume~5 of {\em LIPIcs}, pages
  561--572, 2010.

\bibitem{matching_cut_moshi}
A.~M. Moshi.
\newblock Matching cutsets in graphs.
\newblock {\em Journal of Graph Theory}, 13(5):527--536, 1989.

\bibitem{matching_cut_series_parallel}
M.~Patrignani and M.~Pizzonia.
\newblock The complexity of the matching-cut problem.
\newblock In {\em Proc. of the 27th International Workshop on Graph-Theoretic
  Concepts in Computer Science (WG)}, volume 2204 of {\em LNCS}, pages
  284--295, 2001.

\bibitem{treewidth}
N.~Robertson and P.~Seymour.
\newblock {Graph minors. II. Algorithmic aspects of tree-width}.
\newblock {\em Journal of Algorithms}, 7(3):309--322, 1986.

\bibitem{internal_partition_regular3_4}
K.~H. Shafique and R.~D. Dutton.
\newblock On satisfactory partitioning of graphs.
\newblock {\em Congressus Numerantium}, pages 183--194, 2002.

\bibitem{internal_partition_stiebitz}
M.~Stiebitz.
\newblock Decomposing graphs under degree constraints.
\newblock {\em Journal of Graph Theory}, 23(3):321--324, 1996.

\bibitem{internal_partition_thomassen}
C.~Thomassen.
\newblock Graph decomposition with constraints on the connectivity and minimum
  degree.
\newblock {\em Journal of Graph Theory}, 7(2):165--167, 1983.

\bibitem{Yap83}
C.~Yap.
\newblock Some consequences of non-uniform conditions on uniform classes.
\newblock {\em Theoretical Computer Science}, 26:287--300, 1983.

\end{thebibliography}
\end{document}